\documentclass[12pt]{article}

\usepackage{preamble_paper}

\title{Persuasive Selection in Signaling Games\thanks{First draft: Oct 2025. I thank my advisor, Marek Pycia, for his guidance and support throughout this project. For their comments, I thank Christian Ewerhart, Jindi Huang, Nick Netzer, Armin Schmutzler, Joel Watson, and talk audiences at the University of Zurich, the 9th Swiss Theory Day. All errors are my own.}}
\author{Haoyuan Zeng\thanks{University of Zurich. Email: haoyuan.zeng@econ.uzh.ch.}}

\date{\today \\ \href{https://haoyuanzeng.github.io/papers/PS-Zeng.pdf}{Latest version here.}}

\begin{document}

\maketitle

\begin{abstract}
    This paper introduces a novel criterion, persuasiveness, to select equilibria in signaling games. In response to the Stiglitz critique, persuasiveness focuses on the comparison across equilibria. An equilibrium is more persuasive than an alternative if the set of types of the sender who prefer the alternative would sequentially deviate to the former once other types have done so---that is, if an unraveling occurs. Persuasiveness has strong selective power: it uniquely selects an equilibrium outcome in monotone signaling games. Moreover, in non-monotone signaling games, persuasiveness refines predictions beyond existing selection criteria. Notably, it can also select equilibria in cheap-talk games, where standard equilibrium refinements for signaling games have no selective power.
\end{abstract}


\JELcodes{C70, C72, D82}

\newpage

\section{Introduction}

Many economic interactions can be modeled as a signaling game, where an informed sender sends a message to an uninformed receiver. The receiver responds by taking an action that is payoff-relevant to both players. The seminal paper on signaling games is \citet{spenceJobMarketSignaling1973}'s job market signaling model. Since then, a vast body of literature has applied signaling games across a wide range of fields, including advertising, bargaining, finance, industrial organization, and reputation.\footnote{For a comprehensive survey of the signaling games literature, see \citet{rileySilverSignalsTwentyFive2001} and \citet{sobelSignalingGames2009}. Applications of signaling games span several fields: in advertising, see \citet{nelsonAdvertisingInformation1974}; in bargaining, see \citet{fudenbergSequentialBargainingIncomplete1983} and \citet{sobelMultistageModelBargaining1983}; in finance, see \citet{lelandInformationalAsymmetriesFinancial1977} and \citet{johnDividendsDilutionTaxes1985}; in industrial organization, see \citet{milgromLimitPricingEntry1982}; and in reputation, see \citet{barroRulesDiscretionReputation1983}.}

It is well-known that signaling games often lead to many sequential equilibria \citep{krepsSequentialEquilibria1982}. The multiplicity of equilibria limits the usefulness of the model in analyzing the underlying economic problem, as it fails to yield precise predictions regarding the outcome of strategic interactions. Moreover, for some problems at least, many equilibria seem implausible due to the beliefs associated with messages off the equilibrium path, i.e., off-path beliefs.

A number of equilibrium refinements have been proposed to address the issue of equilibrium multiplicity in signaling games. The classic approach to refining equilibria in signaling games is to formalize plausible restrictions on off-path beliefs, such as the intuitive criterion \citep{choSignalingGamesStable1987} and the divinity criterion \citep{banksEquilibriumSelectionSignaling1987}, which are based on the concepts of strategic stability and forward induction \citep{kohlbergStrategicStabilityEquilibria1986}. The idea is that if sending an off-path message can be interpreted as a signal of certain types of the sender who would be better off by sending this message than following the proposed equilibrium, then this equilibrium fails to pass the equilibrium selection criterion. It is well recognized that this type of reasoning is imperfect.
\begin{quote}
    \textit{...the justification for these restrictions are that they are more ``intuitive''; they go with stories about how players might reason at off-path information sets. I put scare quotes around ``intuitive,'' because what is and isn't intuitive is largely subjective; it is your job to judge which of these restrictions appeals to your intuition and which do not.} \citep[p.~647]{krepsMicroeconomicFoundationsII2023}
\end{quote}

In particular, it has been criticized by Joseph Stiglitz \citep[p.~202]{choSignalingGamesStable1987} for logical inconsistencies stemming from restricting off-path beliefs while holding on-path beliefs fixed. Such inconsistencies lead to the fact that the intuitive criterion always selects the Pareto-dominant separating equilibrium outcome (the Riley outcome \citep{rileyInformationalEquilibrium1979}) in a two-type Spencian model, regardless of how unlikely it is that the worker is of low type. In other words, even when asymmetric information is nearly absent from the firm's perspective---because the low-type worker is extremely rare---the high-type worker nevertheless separates by attaining a high level of education. This separating outcome is substantially worse than the pooling equilibrium outcome, in which the firm would treat them almost as if they were high-type without requiring any costly education. We elaborate on why the intuitive criterion eliminates the pooling equilibrium in Example~\ref{ex:stiglitz} below, where we discuss the Stiglitz critique in detail.

The Stiglitz critique points out that the meaning of messages (beliefs)---whether on-path or off-path---should always be interpreted in equilibrium. As a result, imposing restrictions on off-path beliefs generally also affects on-path beliefs. When both beliefs are consistent with each other, they constitute equilibrium beliefs in an alternative equilibrium. Thus, equilibrium selection reduces to a comparison across alternative equilibria. Based on this idea, \citet*{mailathBeliefBasedRefinementsSignalling1993} introduce the concept of undefeated equilibrium, which is immune to the Stiglitz critique. Nevertheless, in monotone signaling games, such as the job market signaling model \citep{spenceJobMarketSignaling1973}, their criterion typically fails to select a unique equilibrium outcome.

This paper introduces a novel criterion to select equilibria in signaling games, which is called \emph{persuasiveness}. In response to the Stiglitz critique, persuasiveness focuses on the comparison across equilibria. Because different types of the sender may obtain different payoffs in distinct equilibria, they may prefer different equilibria. Accordingly, it is reasonable for the receiver to expect that each type of the sender would like to play the equilibrium which gives them a higher payoff. Persuasiveness formalizes this logic of forward induction by considering how the receiver interprets messages across different equilibria.

To illustrate, consider the two-type Spencian model introduced earlier. When the low-type worker is sufficiently unlikely, both types of the worker prefer the pooling equilibrium to the separating equilibrium, because it yields a higher equilibrium payoff for each type of the worker. Hence, it is reasonable for the firm to expect that both types of the worker would like to play the pooling equilibrium rather than the separating equilibrium. Upon observing a zero education level, the firm therefore finds it more \emph{persuasive} to interpret this message according to the pooling equilibrium---believing that both types choose zero education---than according to the separating equilibrium, in which only the low-type worker does so. In this sense, the pooling equilibrium is more \emph{persuasive} than the separating equilibrium.

More generally, let \(\sigma\) be the putative equilibrium. The problem of equilibrium selection can be conceptualized as the receiver posing the following question to themselves: is there a message \(\overline{m}\) on-path of an alternative equilibrium \(\overline{\sigma}\) for which it is more \emph{persuasive} to interpret \(\overline{m}\) in \(\overline{\sigma}\) than what \(\sigma\) prescribes? Interpreting \(\overline{m}\) in the context of \(\overline{\sigma}\) follows directly from the Stiglitz critique, which emphasizes that the meaning of a message should always be interpreted in equilibrium. In the two-type Spencian model, where the low-type worker is sufficiently unlikely, \(\sigma\) corresponds to the separating equilibrium, \(\overline{\sigma}\) to the pooling equilibrium, and \(\overline{m} = 0\) to the zero education level.

In order for \(\overline{\sigma}\) to serve as a challenger to \(\sigma\), there must exist at least one type of the sender who strictly prefers \(\overline{\sigma}\) to \(\sigma\)---that is, whose equilibrium payoff in \(\overline{\sigma}\) is strictly higher than in \(\sigma\). In the absence of such a type, every type of the sender who sends \(\overline{m}\) in \(\overline{\sigma}\) would instead prefer \(\sigma\) to \(\overline{\sigma}\), implying that the receiver should not expect \(\overline{m}\) to be sent in \(\overline{\sigma}\). When such a type exists, the receiver reasons about the set of types of the sender who send \(\overline{m}\) in \(\overline{\sigma}\) by dividing them into two groups. The first group consists of the set of types of the sender who prefer the challenger \(\overline{\sigma}\) to the putative equilibrium \(\sigma\), while the second group consists of the set of types of the sender who strictly prefer \(\sigma\) to \(\overline{\sigma}\). 

When the second group is empty, \(\overline{\sigma}\) is more \emph{persuasive} than \(\sigma\) in the sense that every type of the sender who sends \(\overline{m}\) in \(\overline{\sigma}\) prefers \(\overline{\sigma}\) to \(\sigma\), and the receiver should expect them to do so. In the two-type Spencian model, where the low-type worker is sufficiently unlikely, the pooling equilibrium \(\overline{\sigma}\) can challenge the separating equilibrium \(\sigma\) with an alternative interpretation of the zero education level \(\overline{m} = 0\)---namely, that this message \(\overline{m}\) is sent by both types rather than solely by the low type. In this scenario, both types of the worker belong to the first group, while the second group is empty. The pooling equilibrium is therefore more \emph{persuasive} than the separating equilibrium, as both types of the worker prefer the pooling equilibrium to the separating equilibrium.

When the second group is non-empty (see Example~\ref{ex:undefeated} below), we check whether the interpretation of \(\overline{m}\) in \(\overline{\sigma}\) rather than in \(\sigma\) provides a rationale for why every type of the sender in the second group would ultimately prefer to play the challenger \(\overline{\sigma}\) rather than the putative equilibrium \(\sigma\), even though they initially prefer \(\sigma\). Consider the simple case where the second group contains only one type \(t'\). When the receiver observes a message \(m'\) sent by type \(t'\) in \(\sigma\), the receiver expects that this message is not sent by the set of types of the sender in the first group, because they would prefer to send \(\overline{m}\) in \(\overline{\sigma}\) instead. If, relative to the equilibrium payoff that type \(t'\) would obtain in the challenger \(\overline{\sigma}\), type \(t'\) is worse off when sending the message \(m'\) in \(\sigma\)---given that the receiver best responds to \(m'\) under the conditional belief that excludes all types in the first group---then type \(t'\) will also have an incentive to deviate to \(\overline{\sigma}\), even though they initially prefer \(\sigma\). In this case, the second group \emph{unravels} as a result of the first group's deviation.

When the second group contains multiple types, we examine them sequentially according to a specified ranking. For each type \(t'\) in this group, we compare their equilibrium payoff in the challenger \(\overline{\sigma}\) with the payoff they would obtain by sending some message \(\tilde{m}'\) in \(\sigma\), assuming the receiver best responds to \(\tilde{m}'\) under the conditional belief that excludes all types in the first group, as well as all types ranked above \(t'\). If, under this belief, type, type \(t'\) is worse off in \(\sigma\) than in \(\overline{\sigma}\), then type \(t'\) will also have an incentive to deviate to \(\overline{\sigma}\). In this case, we say that the interpretation of \(\overline{m}\) in \(\overline{\sigma}\) rather than in \(\sigma\) triggers an \emph{unraveling} (Definition~\ref{def:unraveling}).

The unraveling starts with the highest-ranked type, who finds it better off deviating to \(\overline{\sigma}\) than playing \(\sigma\) if the receiver believes only the second group would play \(\sigma\). Conditional on higher-ranked types having deviated, the unraveling then proceeds to the highest-ranked undeviated type, who similarly prefers deviating to \(\overline{\sigma}\) if the receiver believes only the remaining undeviated types in the second group would play \(\sigma\). This iterative argument continues until all types in the group have deviated. The reasoning is similar to the classical analysis of voluntary disclosure in \citet{grossmanDisclosureLawsTakeover1980}, \citet{grossmanInformationalRoleWarranties1981}, \citet{milgromGoodNewsBad1981}, and \citet{verrecchiaDiscretionaryDisclosure1983}.

Unraveling provides the receiver with a rationale for why every type of the sender who sends \(\overline{m}\) in \(\overline{\sigma}\) would ultimately prefer to play \(\overline{\sigma}\) rather than \(\sigma\), regardless of their initial preference. It is thus more persuasive to interpret \(\overline{m}\) in \(\overline{\sigma}\) rather than in \(\sigma\). Note that we have a trivial unraveling when the second group is empty. Hence, we say that \(\overline{\sigma}\) is \emph{more persuasive} than \(\sigma\) if there exists a message \(\overline{m}\) on the equilibrium path of \(\overline{\sigma}\) such that the interpretation of \(\overline{m}\) in \(\overline{\sigma}\) rather than in \(\sigma\) triggers an unraveling.

An equilibrium \(\overline{\sigma}\) is \emph{most persuasive} if \(\overline{\sigma}\) is more persuasive than any other equilibrium \(\sigma\) that is not payoff-equivalent for the sender. In other words, the most persuasive equilibrium \(\overline{\sigma}\) can always challenge any other equilibrium with a new interpretation of some message such that every type of the sender who sends that message in \(\overline{\sigma}\) would like to deviate to that message and play \(\overline{\sigma}\). If the most persuasive equilibrium is unique up to payoff equivalence for the sender, then no other equilibrium can challenge the most persuasive equilibrium with new interpretations of any messages. Hence, when there exists a unique most persuasive equilibrium, the interpretations of all messages in the game are determined, in the sense that no other equilibrium can provide a more persuasive interpretation of any message.

Unlike previous equilibrium refinements, our selection criterion emphasizes how the receiver interprets messages rather than how the sender signals through off-path messages. In particular, it does not require the receiver to detect that the sender is deviating from \(\sigma\) by sending an off-path message. Instead, by introspection, the receiver finds it more persuasive to interpret \(\overline{m}\) in \(\overline{\sigma}\) rather than in \(\sigma\). For instance, in the two-type Spencian model, the zero education level is on-path of the separating equilibrium. By introspection, the firm finds it more persuasive to interpret this message as being sent by both types when the low-type worker is sufficiently unlikely. Persuasiveness therefore has selective power even when \(\overline{m}\) is on the equilibrium path of \(\sigma\) (see also Example~\ref{ex:discrete} below). The logic of forward induction underlying persuasiveness is that the receiver expects the sender would like to play the equilibrium which gives them a higher payoff---a natural assumption given that the sender moves first.

Persuasiveness has strong selective power. In monotone signaling games, such as the job market signaling model \citep{spenceJobMarketSignaling1973}, it uniquely selects the equilibrium outcome that provides the highest type of the sender with the maximum equilibrium payoff (lexicographically maximum outcome). Next, we illustrate by examples that persuasiveness has stronger selective power than other equilibrium refinements in some non-monotone signaling games (see Table~\ref{tab:comparison} for a summary of the comparison). We also discuss the limitations of persuasiveness, noting that cyclicality can emerge in the absence of a unique most persuasive equilibrium. We explain why it may be reasonable to consider the \emph{least persuasive} equilibrium in some cases. Lastly, we demonstrate that persuasiveness can have good selective power even in cheap-talk games, where standard equilibrium refinements for signaling games have no selective power.

The paper is organized as follows. Section~\ref{sec:setup} introduces the setup. Section~\ref{sec:persuasive} discusses the Stiglitz critique and formally defines persuasiveness. Section~\ref{sec:monotone-signaling-games} studies persuasiveness in monotone signaling games. Section~\ref{sec:discussion} discusses the extensions and limitations of our selection criterion. Section~\ref{sec:related-literature} reviews the related literature. Section~\ref{sec:conclusion} concludes the paper. All proofs are relegated to Appendix~\ref{sec:proofs}. In Appendix~\ref{sec:refinements}, we offer an intuitive explanation of some equilibrium refinements in the previous literature.

\section{Setup}\label{sec:setup}

The signaling game $G$ is described as follows. There is a sender (\(S\)) and a receiver (\(R\)). It is common knowledge that Nature draws the type \(t\) of the sender from a non-empty set of types $T$ according to a prior probability distribution $p \in \Delta\left(T\right)$ with full support. The sender privately observes \(t\) and chooses a message $m$ from a non-empty set of messages $M$. The receiver only observes the message \(m\) and chooses an action $a$ from a non-empty set of actions $A$. Both the sender's and the receiver's payoffs depend on the message, the action, and the sender's type, i.e., $u_{S},u_{R}:T\times M\times A\to\mathbb{R}$.

A strategy for the sender \(\sigma_{S}\) maps types to distributions over messages, i.e., $\sigma_{S}:T\to\Delta\left(M\right)$. Let $\sigma_{S}\left(\left.m\right|t\right)$ denote the probability of the type $t$ sender choosing the message $m$. Upon receiving the sender's message, the receiver updates their belief \(\mu\) over the sender's types, i.e., $\mu:M\to\Delta\left(T\right)$. Let $\mu\left(\left.t\right|m\right)$ denote the receiver's posterior belief about the sender being type $t$ after receiving the message $m$. A strategy for the receiver \(\sigma_{R}\) maps messages to distributions over actions, i.e., $\sigma_{R}:M\to\Delta\left(A\right)$. Let $\sigma_{R}\left(\left.a\right|m\right)$ denote the probability of action $a$ being played upon receiving the message $m$. Let $\text{BR}\left(m,\mu\right)$ be the set of actions that are best responses to the message $m$ given the belief $\mu$. That is,
\[
\text{BR}\left(m,\mu\right)=\arg\max_{a\in A}\sum_{t\in T}u_{R}\left(t,m,a\right)\mu\left(\left.t\right|m\right).
\]
Let \(\text{supp}\left(\cdot\right)\) be the support of a function, i.e., the set of points at which the function is non-zero. For example, if $m \in \text{supp}\left(\sigma_{S}\left(t\right)\right)$, then $\sigma_{S}\left(\left.m\right|t\right) > 0$.

A sequential equilibrium $\sigma$ is a strategy-belief profile
\[
\sigma = \left(\sigma_{S}\left(\left.m\right|t\right),\sigma_{R}\left(\left.a\right|m\right),\mu\left(\left.t\right|m\right)\right)_{t \in T, m\in M}
\]
that satisfies sequential rationality and consistency.\footnote{Since the signaling game is a two-period game with observed actions and independent types, any perfect Bayesian equilibrium is also a sequential equilibrium \citep{fudenbergPerfectBayesianEquilibrium1991}.}
\begin{definition}
    $\sigma=\left(\sigma_{S},\sigma_{R},\mu\right)$ is a \emph{sequential equilibrium} if,
    \begin{enumerate}
        \item Sequential Rationality: $\forall t \in T$,
        \[
        \text{supp}\left(\sigma_{S}\left(t\right)\right)\subseteq\arg\max_{m\in M}\sum_{a\in A}u_{S}\left(t,m,a\right)\sigma_{R}\left(\left.a\right|m\right),
        \]
        where $\text{supp}\left(\sigma_{R}\left(m\right)\right)\subseteq\text{BR}\left(m,\mu\right) \ \forall m \in M$.
        \item Consistency: $\forall t \in T$, $\forall m\in M$,
        \[
        \mu\left(\left.t\right|m\right)=\frac{\sigma_{S}\left(\left.m\right|t\right)p\left(t\right)}{\sum_{t' \in T}\sigma_{S}\left(\left.m\right|t'\right)p\left(t'\right)}
        \]
        conditional on $\sum_{t' \in T}\sigma_{S}\left(\left.m\right|t'\right)p\left(t'\right)>0$.
\end{enumerate}
\end{definition}

Let $\text{SE}\left(G\right)$ denote the set of sequential equilibria in the game $G$. Let $\text{PSE}\left(G\right)$ denote the set of pure-strategy sequential equilibria in the game $G$. Let $u_{S}\left(t,\sigma\right)$ denote the expected payoff of the type $t$ sender in the equilibrium $\sigma$.

A message \(m\) is off the equilibrium path of \(\sigma\) if \(\sum_{t \in T} \sigma_{S}\left(\left.m\right|t\right)p\left(t\right) = 0\), i.e., no type of the sender sends \(m\) with positive probability in \(\sigma\). We also call \(m\) an off-path message (with respect to \(\sigma\)). The belief \(\mu\left(\cdot|m\right)\) of the receiver after observing \(m\) is called an off-path belief if \(m\) is off the equilibrium path of \(\sigma\).

\section{Persuasive Selection}\label{sec:persuasive}

In this section, we start with the intuitive criterion \citep{choSignalingGamesStable1987}, one of the most widely used equilibrium refinements in signaling games. Using Example~\ref{ex:stiglitz}, we illustrate how it is subject to the Stiglitz critique---a limitation that also applies to related selection criteria \citep{banksEquilibriumSelectionSignaling1987,grossmanPerfectSequentialEquilibrium1986}. Then, we introduce the concept of persuasiveness, which is immune to the Stiglitz critique.

The intuitive criterion is based on the idea of forward induction \citep{kohlbergStrategicStabilityEquilibria1986}: if sending an off-path message can be interpreted as a signal of certain types of the sender who would be better off by sending this message than following the proposed equilibrium, then this equilibrium fails to pass the intuitive criterion. Verifying whether an equilibrium \(\sigma\) fails the intuitive criterion involves the following two steps.\footnote{This explanation of the intuitive criterion follows from \citet{munoz-garciaIntuitiveDivinityCriterion2011}.}

\begin{itemize}
    \item Step 1: Which types of the sender \emph{could benefit} by sending an off-path message \(m\)?
    
    We denote the set of such types as \(D\). Formally,
    \[
    D = \left\{\left. t \in T \right| u_{S}\left(t, \sigma\right) \leq \max_{a \in \text{BR}\left(m, \Delta\left(T\right)\right)} u_{S}\left(t, m, a\right)\right\},
    \]
    where \(\text{BR}\left(m, \Delta\left(T\right)\right) = \cup_{\mu \in \Delta\left(T\right)}\text{BR}\left(m, \mu\right)\).

    \item Step 2: If deviations only come from the set of types of the sender identified in Step 1, is the \emph{lowest} payoff from deviating higher than their equilibrium payoff for some type of the sender?

    Formally, if there exists \(t \in D\) such that
    \[
    \min_{a \in \text{BR}\left(m, \Delta\left(D\right)\right)} u_{S}\left(t, m, a\right) > u_{S}\left(t, \sigma\right),
    \]
    then this equilibrium \(\sigma\) fails the intuitive criterion.
\end{itemize}

The reader is referred to \citet{choSignalingGamesStable1987} for a detailed discussion of the intuitive criterion. The authors are well aware of the subtleties involved.
\begin{quote}
    \textit{``Despite the name we have given it, the Intuitive Criterion is not completely intuitive.''} \citep[p.~202]{choSignalingGamesStable1987}
\end{quote}

Here, we highlight that Step 1 performs a \emph{best-case} scenario analysis for the sender, while Step 2 performs a \emph{worst-case} scenario analysis. In Step 1, we choose the \emph{best} action for the sender as long as the receiver's action is a best response to some belief---equivalently, we select the \emph{most favorable} receiver's belief for the sender. In Step 2, by contrast, we choose the \emph{worst} action for the sender---equivalently, we select the \emph{least favorable} receiver's belief for the sender. Consequently, beliefs jump from one extreme to the other when moving from Step 1 to Step 2. This asymmetry in the treatment of beliefs underlies the Stiglitz critique, as illustrated in Example~\ref{ex:stiglitz} below. As we will see later, it also accounts for the intuitive criterion's weak selective power among pooling equilibria (see below Example~\ref{ex:hiding}).

\subsection{The Stiglitz Critique}\label{sec:stiglitz-critique}

We consider a simple two-type version of \citet{spenceJobMarketSignaling1973}'s job market signaling model.

\begin{example}[Two-Type Spencian Game]\label{ex:stiglitz}
    There are two types of a worker (sender), \(T = \left\{t_{L}, t_{H}\right\}\), where \(t_{L} = 1\) and \(t_{H} = 2\). We call \(t_{L}\) the low-type worker, and \(t_{H}\) the high-type worker. The prior probability that the worker is low-type is \(p \in \left(0,1\right)\). The worker chooses an education level \(m \in M = [0,\infty)\). The cost of acquiring education is \(c\left(t, m\right) = \frac{m}{t}\). A firm (receiver) observes the education level of the worker and offers a wage of \(a \in A = \mathbb{R}_{+}\) in a competitive market. In equilibrium, the firm always earns zero profit, i.e., \(\sigma_{R}\left(m\right) = \mathbb{E}_{\mu}\left[\left.t\right|m\right]\).\footnote{Strictly speaking, we should model this as a game with two firms, who would then engage in a Bertrand competition for the worker. A single firm with the payoff function \(u_{R}\left(t, m, a\right) =-\left(t-a\right)^{2}\) yields the similar behavior.} The payoff of a worker of type \(t\) who chooses an education level \(m\) and receives a wage of \(a\) is given by \(u_{S}\left(t,m,a\right) = a - c\left(t,m\right)\).

    It is well-known that this game has a Pareto-dominating separating equilibrium, i.e., the Riley equilibrium \(\sigma^{\text{Riley}}\) \citep{rileyInformationalEquilibrium1979}, in which the low-type worker chooses \(m=0\) and the high-type worker chooses \(m=1\). The firm offers a wage of \(a=1\) to a worker with \(m=0\) and \(a=2\) to a worker with \(m=1\). The payoffs of the low-type and the high-type workers in this equilibrium are \(u_{S}\left(t_{L},\sigma^{\text{Riley}}\right) = 1\) and \(u_{S}\left(t_{H},\sigma^{\text{Riley}}\right) = 1.5\), respectively.

    This game also has a pooling equilibrium \(\sigma^{\text{Pooling}}\), in which both types of workers choose \(m=0\) and the firm offers an expected wage of \(a=\mathbb{E}_{\mu^{\text{Pooling}}}\left[\left.t\right|m=0\right] = 2 - p\). The payoffs of the low-type and the high-type workers in this equilibrium are \(u_{S}\left(t_{L},\sigma^{\text{Pooling}}\right) =u_{S}\left(t_{H},\sigma^{\text{Pooling}}\right) = 2 - p\).\footnote{There exist other equilibria, but they are not relevant for our discussion here because they are generally deemed implausible by any equilibrium refinement.}

    For both equilibria, after an off-path message, the firm believes that the worker is low-type with probability one. We can summarize the equilibrium strategies and payoffs of the low-type and high-type workers in the two equilibria in Table~\ref{tab:stiglitz}.
    \begin{table}[ht!]
        \centering
        \bgroup
        \def\arraystretch{1.5}
        \begin{tabular}{ccc}
        \toprule
            & \(t_L\) & \(t_H\) \\
        \midrule
        \(\sigma_{S}^{\text{Riley}}\left(t\right)\)    &  0 & 1  \\
        \(\sigma_{S}^{\text{Pooling}}\left(t\right)\)    & 0 & 0 \\
        \bottomrule
        \end{tabular}
        \quad
        \begin{tabular}{ccc}
        \toprule
            & \(t_L\) & \(t_H\) \\
        \midrule
        \(u_{S}\left(t,\sigma^{\text{Riley}}\right)\)    &  1 & 1.5  \\
        \(u_{S}\left(t,\sigma^{\text{Pooling}}\right)\)    & \(2-p\) & \(2-p\)  \\
        \bottomrule
        \end{tabular}
        \egroup
        \caption{Equilibrium Strategies and Payoffs of the Worker in Example~\ref{ex:stiglitz}\label{tab:stiglitz}}
    \end{table}
\end{example}

The intuitive criterion uniquely selects the Riley outcome in \(\sigma^{\text{Riley}}\) irrespective of \(p\) \citep{choSignalingGamesStable1987}.\footnote{The intuitive criterion uniquely selects the Riley outcome instead of the Riley equilibrium \(\sigma^{\text{Riley}}\) itself. There can be multiple Riley equilibria which produce exactly the same Riley outcome on path. However, they can differ on the off-path beliefs, which cannot be uniquely pinned down in general.} Notice that when \(p\) is close to zero, i.e., the low-type worker is very unlikely, the pooling outcome in \(\sigma^{\text{Pooling}}\) Pareto-dominates the Riley outcome in \(\sigma^{\text{Riley}}\) for each type of the worker. Intuitively, when \(p\) is close to zero, the firm offers a wage close to \(2\) in the pooling equilibrium. The high-type worker does not want to incur a cost to separate themselves from the low-type worker, because it is almost the same as if there were no low-type worker at all. In particular, when \(p = 0\), the pooling outcome degenerates to the equilibrium outcome when there is no low-type worker.\footnote{In contrast, the Riley outcome stays the same as long as \(p>0\). When \(p=0\), the equilibrium payoff of the high-type worker jumps from 1.5 to 2. There is generally a discontinuity in the Riley outcome as one of the prior probabilities goes to zero for any finite type space \citep*{mailathBeliefBasedRefinementsSignalling1993}.} It seems counterintuitive that the pooling equilibrium fails the intuitive criterion, even though its outcome Pareto-dominates the Riley outcome for each type of the worker. This raises the question of why the intuitive criterion rejects the pooling equilibrium when \(p\) approaches zero.

To see this, consider an off-path education level \(m' \in (p, 2p)\). In Step 1 of the intuitive criterion, the firm would at most offer a wage of 2 after observing \(m'\). The payoff of the low-type worker if deviating to \(m'\) is at most \(2 - m' < 2 - p\), while the payoff of the high-type worker if deviating to \(m'\) is at most \(2 - \frac{m'}{2} > 2 - p\). Then, only the high-type worker could benefit by deviating to \(m'\). In Step 2, since the firm believes that only the high-type worker would choose \(m'\), the firm would offer a wage of 2 and the high-type worker would benefit from this deviation. Hence, the pooling equilibrium fails the intuitive criterion.

The Stiglitz critique points out that the above reasoning is flawed. The two steps of the intuitive criterion are usually motivated as the high-type worker making an implicit ``speech'' to the firm by deviating to \(m'\). However, this ``speech'' induces an inconsistent ``story'' when taking into account the low-type worker. Following the logic of the intuitive criterion, the firm would now believe that the worker is low-type with probability one after observing \(m = 0\) because the high-type worker would have deviated to \(m'\). Then, the low-type worker would not keep choosing \(m = 0\) as in the pooling equilibrium because the firm would only offer a wage of 1 after observing \(m = 0\). Instead, the low-type worker has an incentive to deviate to \(m'\) and and imitate the high-type worker, as this deviation yields a payoff of \(2 - m' > 1\) when \(p\) is close to zero. If the firm anticipates the low-type worker's response after the high-type worker's deviation, the firm would offer an expected wage of \(2-p\) instead of \(2\). As a result, the high-type worker would not want to deviate to \(m'\), because they incur a cost by choosing \(m' > 0\) but the wage is the same as before in the pooling equilibrium, which invalidates the reason why the pooling equilibrium fails the intuitive criterion.

These logical inconsistencies in the intuitive criterion stem from restricting off-path beliefs while holding on-path beliefs fixed. In Example~\ref{ex:stiglitz}, the intuitive criterion postulates that the off-path belief of the firm after observing \(m'\) is that the worker is high-type with probability one. However, it also implicitly assumes that the on-path belief of the firm after observing \(m=0\) is that the worker is high-type with probability \(1-p\). Apparently, these two beliefs cannot hold simultaneously. This implies that when we impose restrictions on off-path beliefs, consistency requires further adjustments on on-path beliefs. When these beliefs are consistent with each other, we are effectively looking at another equilibrium. Hence, to address the Stiglitz critique, we should always interpret the meaning of messages (beliefs after observing \(m=0\) or \(m'\))---whether on-path or off-path---in equilibrium. Equilibrium selection concerns identifying the most appropriate interpretation of a given message among multiple equilibria.

In Example~\ref{ex:stiglitz}, when the firm interprets \(m'\) in equilibrium, the firm can at most offer an expected wage of \(2-p\).\footnote{This wage is consistent with another pooling equilibrium at \(m'\).} Hence, the high-type worker would not benefit from deviating to \(m'\), and strictly prefers to play the pooling equilibrium. Given this, how should we conceptualize the selection between the pooling equilibrium and the Riley equilibrium?

Note that the low-type worker always chooses \(m=0\) in both equilibria. Then, the high-type worker is effectively facing the decision of whether to separate from the low-type worker by choosing \(m=1\) or to pool with the low-type worker by choosing \(m=0\). When \(p > 0.5\), the high-type worker strictly prefers separation to pooling. When \(p \leq 0.5\), the high-type worker weakly prefers pooling to separation.\footnote{When \(p=0.5\), the high-type worker is indifferent between separation and pooling, but the low-type worker always strictly prefers pooling to separation. In such a case, we select the weakly Pareto-dominating equilibrium, i.e., the pooling equilibrium.} Since the worker moves first, it is reasonable for the firm to expect the high-type worker to go with the choice that gives them a higher payoff. As a result, when \(p > 0.5\), the firm should interpret \(m=0\) in the Riley equilibrium rather than in the pooling equilibrium; when \(p \leq 0.5\), the firm should interpret \(m=0\) in the pooling equilibrium rather than in the Riley equilibrium. Since \(m = 0\) is on-path of both equilibria, the reasoning is not about how the high-type worker signals through off-path education levels as in the intuitive criterion but rather how the firm expects what the high-type worker would do. The interpretation of \(m=0\) changes with the equilibrium.

The above reasoning motivates a selection criterion that neither selects the pooling outcome nor the Riley outcome irrespective of \(p\). Instead, we select the equilibrium outcome that gives the high-type worker a higher payoff, which we will call the lexicographically maximum outcome (lex max outcome) later (Definition~\ref{def:LMSE}). We call the equilibrium that produces this outcome the lexicographically maximum sequential equilibrium (LMSE). In Example~\ref{ex:stiglitz}, the LMSE is the Riley equilibrium when \(p > 0.5\) and the pooling equilibrium when \(p \leq 0.5\). To address the Stiglitz critique, the selection criterion should build on how the receiver interprets messages in different equilibria and finding a good interpretation among those, which will be formalized as the concept of \emph{persuasiveness} now.

\subsection{Persuasiveness}\label{sec:persuasiveness}

To illustrate the logic of persuasiveness, we consider a three-type version of \citet{spenceJobMarketSignaling1973}'s job market signaling model used in \citet*{mailathBeliefBasedRefinementsSignalling1993}.

\begin{example}[Three-Type Spencian Game]\label{ex:undefeated}
    We follow the setup of Example~\ref{ex:stiglitz} except that there are three types of a worker, \(T = \left\{t_{L}, t_{M}, t_{H}\right\}\), where \(t_{L} = 1\), \(t_{M} = 2\), and \(t_{H} = 3\). We call \(t_{L}\) the low-type worker, \(t_{M}\) the medium-type worker, and \(t_{H}\) the high-type worker. The prior probabilities are \(p\left(t_{L}\right) = 0.35\), \(p\left(t_{M}\right) = 0.20\), and \(p\left(t_{H}\right) = 0.45\).

    We focus on the following equilibria with different information revealed, and summarize in Table~\ref{tab:undefeated} below the equilibrium strategies and payoffs of the low-type, medium-type, and high-type workers in these equilibria, presented in the same format as Example~\ref{ex:stiglitz}.\footnote{To make the comparison easier, we round numbers to two decimal places when necessary.} For instance, in the equilibrium \(\sigma^{1}\), the low-type and medium-type workers choose \(m=0\), while the high-type worker chooses \(m = 3.27\), which is the minimum level of education the high-type worker has to choose in order to separate themselves from the low-type and medium-type workers. The firm offers an expected wage of \(a= \mathbb{E}_{\mu^{1}}\left[\left.t\right|m=0\right] = \frac{0.35 \times 1 + 0.2 \times 2}{0.35 + 0.2} = 1.36\) to a worker with \(m=0\) and \(a=3\) to a worker with \(m=3.27\). The payoffs of the low-type, medium-type, and high-type workers in this equilibrium are \(u_{S}\left(t_{L},\sigma^{1}\right) = u_{S}\left(t_{M},\sigma^{1}\right) = 1.36\), and \(u_{S}\left(t_{H},\sigma^{1}\right) = 1.91\), respectively. The other equilibria can be similarly interpreted. In particular, the equilibrium \(\overline{\sigma}\) is the LMSE, where the high-type worker attains the highest equilibrium payoff.

    \begin{table}[ht!]
        \centering
        \bgroup
        \def\arraystretch{1.5}
        \begin{tabular}{cccc}
        \toprule
            & \(t_L\) & \(t_M\) & \(t_H\) \\
        \midrule
        \(\sigma_{S}^{1}\left(t\right)\)    &  0 & 0 & 3.27 \\
        \(\sigma_{S}^{\text{Riley}}\left(t\right)\)    &  0 & 1 & 3 \\
        \(\sigma_{S}^{\text{Pooling}}\left(t\right)\)  &  0 & 0 & 0 \\
        \(\overline{\sigma}_{S}\left(t\right)\)    & 0 & 1.1 & 1.1 \\
        \bottomrule
        \end{tabular}
        \quad
        \begin{tabular}{cccc}
        \toprule
            & \(t_L\) & \(t_M\) & \(t_H\) \\
        \midrule
        
        \(u_{S}\left(t,\sigma^{1}\right)\)    &  1.36 & 1.36 & 1.91  \\
        \(u_{S}\left(t,\sigma^{\text{Riley}}\right)\)    &  1 & 1.5 & 2  \\
        \(u_{S}\left(t,\sigma^{\text{Pooling}}\right)\)  &  2.1 & 2.1 & 2.1  \\
        \(u_{S}\left(t,\overline{\sigma}\right)\)    & 1 & 1.85 & 2.13  \\
        \bottomrule
        \end{tabular}
        \egroup
        \caption{Equilibrium Strategies and Payoffs of the Worker in Example~\ref{ex:undefeated}\label{tab:undefeated}}
    \end{table}
\end{example}

Note that the intuitive criterion has limited selective power when there are more than two types. The D1 criterion is applied more often instead, because it uniquely selects the Riley outcome \citep{choStrategicStabilityUniqueness1990}.\footnote{See \citet{munoz-garciaIntuitiveDivinityCriterion2011} for a detailed explanation of why the intuitive criterion cannot uniquely select the Riley outcome.} Because the D1 criterion is also subject to the Stiglitz critique for the same reason as the intuitive criterion, we see in this example again that the Riley outcome is uniquely selected even when it is Pareto-dominated by the pooling outcome in \(\sigma^{\text{Pooling}}\) for each type of the worker. However, we will not argue that the pooling outcome should be selected instead of the Riley outcome in this case as in Example~\ref{ex:stiglitz}. Instead, we will show that the LMSE \(\overline{\sigma}\) is \emph{more persuasive} than both the Riley equilibrium \(\sigma^{\text{Riley}}\) and the pooling equilibrium \(\sigma^{\text{Pooling}}\), and hence the lex max outcome should be selected.

It is easy to see that the LMSE \(\overline{\sigma}\) Pareto-dominates the Riley equilibrium \(\sigma^{\text{Riley}}\) for the medium-type and high-type workers. Given \(\sigma^{\text{Riley}}\), we consider what happens when the medium-type and high-type workers both deviate to choose \(\overline{m} = 1.1\). When the firm interprets \(\overline{m}\) in \(\overline{\sigma}\), the firm would offer the corresponding equilibrium wage in \(\overline{\sigma}\), and both types of workers would indeed be better off by deviating to \(\overline{m}\) than following the Riley equilibrium. In this sense, we say that \(\overline{\sigma}\) is \emph{more persuasive} than \(\sigma^{\text{Riley}}\).

However, the above reasoning does not work when comparing the LMSE \(\overline{\sigma}\) with the pooling equilibrium \(\sigma^{\text{Pooling}}\), because the high-type worker prefers \(\overline{\sigma}\) to \(\sigma^{\text{Pooling}}\), while the medium-type worker prefers \(\sigma^{\text{Pooling}}\) to \(\overline{\sigma}\). To argue that \(\overline{\sigma}\) is still \emph{more persuasive} than \(\sigma^{\text{Pooling}}\), we first consider what could happen to the medium-type worker when they do not deviate to \(\overline{m}\) and instead continue to send \(m'=0\), as in \(\sigma^{\text{Pooling}}\),  while the high-type worker deviates to \(\overline{m}\).

As in Example~\ref{ex:stiglitz}, the firm expects that the worker would like to play the equilibrium which gives them a higher payoff. Therefore, upon observing \(m' = 0\), the firm would no longer interprets this message as part of the pooling equilibrium, as the high-type worker has deviated to \(\overline{m}\). The firm would infer that a worker choosing \(m' = 0\) must be either medium-type or low-type---precisely as in \(\sigma^{1}\). Consequently, the medium-type worker would obtain the same payoff as in \(\sigma^{1}\) (1.36), which is lower than the payoff obtained in \(\overline{\sigma}\) (1.85) when following the high-type worker and deviating to \(\overline{m}\). Hence, the medium-type worker has no incentive to keep sending \(m' = 0\), since doing so would only reveal that they are not high-type. In other words, although the medium-type worker initially prefers \(\sigma^{\text{Pooling}}\) to \(\overline{\sigma}\), they ultimately find it optimal to deviate to \(\overline{m}\) once the high-type worker has done so. We summarize this reasoning by saying that the interpretation of \(\overline{m}\) in \(\overline{\sigma}\) rather than in \(\sigma^{\text{Pooling}}\) triggers an \emph{unraveling}.

Conversely, we consider what would happen to the high-type worker when they deviate to \(\overline{m}\), while the medium-type worker continues to send \(m' = 0\), as in \(\sigma^{\text{Pooling}}\). Since the high-type worker obtains the highest equilibrium payoff by sending \(\overline{m}\) in \(\overline{\sigma}\), the firm should believe that the worker could be high-type after observing \(\overline{m}\). When the firm interprets \(\overline{m}\) in \(\overline{\sigma}\), the worker could also be medium-type. However, the high-type worker has an incentive to deviate to \(\overline{m}\) regardless of whether the firm believes the medium-type worker would also deviate. In particular, if the firm believes that only the high-type worker would deviate to \(\overline{m}\), the firm would offer a wage higher than that in \(\overline{\sigma}\), making the high-type worker even better off. Hence, the high-type worker's deviation to \(\overline{m}\) is unaffected by the medium-type worker's choice. There will be no unraveling for the high-type worker.

When the firm expects that the worker would like to play the equilibrium that give them a higher payoff, the high-type worker would like to deviate to \(\overline{m}\) and play \(\overline{\sigma}\) because there is no unraveling, while the medium-type worker would ultimately deviate to \(\overline{m}\) and deviate to \(\overline{\sigma}\) because of unraveling. Taken together, unraveling provides a consistent story for the firm, explaining why both the medium-type and high-type workers would like to deviate to \(\overline{m}\) and play \(\overline{\sigma}\) instead of \(\sigma^{\text{Pooling}}\) irrespective of their initial preferences for \(\overline{\sigma}\) or \(\sigma^{\text{Pooling}}\). In this sense, we say that \(\overline{\sigma}\) is \emph{more persuasive} than \(\sigma^{\text{Pooling}}\).

Now we formalize the concept of persuasiveness in any signaling game \(G\). In response to the Stiglitz critique, persuasiveness builds on how the receiver interprets messages in different equilibria. Consider two equilibria \(\overline{\sigma} = \left(\overline{\sigma}_{S},\overline{\sigma}_{R},\overline{\mu}\right), \sigma = \left(\sigma_{S},\sigma_{R},\mu\right)\in \text{SE}\left(G\right)\). We can view \(\overline{\sigma}\) and the set \(\{\sigma^1, \sigma^{\text{Riley}}, \sigma^{\text{Pooling}}\}\) in Example~\ref{ex:undefeated} as specific instances of \(\overline{\sigma}\) and \(\sigma \in \{\sigma^1, \sigma^{\text{Riley}}, \sigma^{\text{Pooling}}\}\). In order for \(\overline{\sigma}\) to serve as a challenger to the putative equilibrium \(\sigma\), we start with a message \(\overline{m}\) that is on-path of the challenger \(\overline{\sigma}\) such that there exists a type \(\overline{t}\) who obtains a strictly higher payoff in \(\overline{\sigma}\) than in \(\sigma\). If no such message exists, every type of the sender will prefer \(\sigma\) to \(\overline{\sigma}\). Then, the receiver should not expect that \(\overline{\sigma}\) is played instead of \(\sigma\). In Example~\ref{ex:undefeated}, we have \(\overline{m} = 1.1\), and \(\overline{t} = t_{H}\).

Following the analysis in Example~\ref{ex:undefeated}, we first divide the set of types of the sender who sends \(\overline{m}\) in \(\overline{\sigma}\) into two groups based on their preferences for \(\overline{\sigma}\) or \(\sigma\).\footnote{When some type of the sender is indifferent, we put them in the group that prefers \(\overline{\sigma}\), which makes persuasiveness more selective (in a reasonable way). For instance, in Example~\ref{ex:stiglitz}, when \(p = 0.5\), the high-type worker is indifferent between separation and pooling. Still, the pooling equilibrium is more persuasive than the Riley equilibrium because the low-type worker strictly prefers pooling.} We allow mixed strategies and define the two groups as follows:
\begin{align*}
    T^{\overline{\sigma} \geq \sigma}_{\overline{m}} &=\left\{ \left.t \in T\right|\overline{m}\in\text{supp}\left(\overline{\sigma}_{S}\left(t\right)\right), u_{S}\left(t,\overline{\sigma}\right) \geq u_{S}\left(t,\sigma\right)\right\}\\
    T_{\overline{m}}^{\overline{\sigma} < \sigma} &=\left\{ \left.t \in T\right|\overline{m}\in\text{supp}\left(\overline{\sigma}_{S}\left(t\right)\right), u_{S}\left(t, \overline{\sigma}\right) < u_{S}\left(t, \sigma\right)\right\}.
\end{align*}

In Example~\ref{ex:undefeated}, when comparing \(\overline{\sigma}\) with either \(\sigma^1\) or \(\sigma^{\text{Riley}}\), we have \(T^{\overline{\sigma} \geq \sigma^{1}}_{\overline{m}} = T^{\overline{\sigma} \geq \sigma^{1}}_{\overline{m}} = \left\{t_{H}, t_{M}\right\}\), and \(T^{\overline{\sigma} < \sigma^{1}}_{\overline{m}} = T^{\overline{\sigma} < \sigma^{\text{Riley}}}_{\overline{m}} = \emptyset\). When comparing \(\overline{\sigma}\) to \(\sigma^{\text{Pooling}}\), we have \(T^{\overline{\sigma} \geq \sigma^{\text{Pooling}}}_{\overline{m}} = \left\{t_{H}\right\}\) and \(T^{\overline{\sigma} < \sigma^{\text{Pooling}}}_{\overline{m}} = \left\{t_{M}\right\}\). Building on the preceding discussion, we characterize the unraveling dynamics in Example~\ref{ex:undefeated} as follows.

\begin{definition}\label{def:unraveling}
    The interpretation of \(\overline{m}\) in \(\overline{\sigma}\) rather than in \(\sigma\) triggers an \emph{unraveling}, if (1) there exists \(\overline{t} \in T^{\overline{\sigma} \geq \sigma}_{\overline{m}}\) such that \(u_{S}\left(\overline{t},\overline{\sigma}\right)>u_{S}\left(\overline{t},\sigma\right)\); and (2) there exists a ranking of types \( f : T^{\overline{\sigma} < \sigma}_{\overline{m}} \to \mathbb{R} \) such that for all \(t' \in T_{\overline{m}}^{\overline{\sigma} < \sigma}\), and all $m'\in\text{supp}\left(\sigma_{S}\left(t'\right)\right)$,
    \begin{equation}\label{eq:unraveling}
    u_{S}\left(t',\overline{\sigma}\right)\geq\max_{a'\in \text{BR}\left(m',\mu'\right)}u_{S}\left(t',m',a'\right),
    \end{equation}
    where
    \begin{align*}
    \mu'\left(\left.t\right|m'\right) & = \begin{cases}
        \frac{\mu\left(\left.t\right|m'\right)}{\sum_{\hat{t} \in U^{\sigma>\overline{\sigma}}_{m'}}\mu\left(\left.\hat{t}\right|m'\right)} & \text{if } t \in U^{\sigma>\overline{\sigma}}_{m'},\\
        0 & \text{otherwise.}
    \end{cases}\\
    U^{\sigma>\overline{\sigma}}_{m'} & = T^{\sigma}_{m'} \setminus \left(F^{\overline{\sigma} < \sigma}_{\overline{m}}\left(t'\right) \cup T^{\overline{\sigma} \geq \sigma}_{\overline{m}}\right) \\
    T^{\sigma}_{m'} & =\left\{ \left.t \in T\right|m'\in\text{supp}\left(\sigma_{S}\left(t\right)\right)\right\} \\
    F^{\overline{\sigma} < \sigma}_{\overline{m}}\left(t'\right) & =\left\{ \left.t \in T^{\overline{\sigma} < \sigma}_{\overline{m}}\right|f\left(t\right) > f\left(t'\right)\right\}.
    \end{align*}
\end{definition}

Note that if \(T^{\overline{\sigma} < \sigma}_{\overline{m}} = \emptyset\), the comparison between \(\overline{\sigma}\) and \(\sigma\) leads to a trivial unraveling. This corresponds to the situation in Example~\ref{ex:undefeated} when comparing \(\overline{\sigma}\) with either \(\sigma^{1}\) or \(\sigma^{\text{Riley}}\). In such cases, all types who send \(\overline{m}\) in \(\overline{\sigma}\) prefer \(\overline{\sigma}\) to \(\sigma\), and hence they all have an incentive to deviate to \(\overline{m}\) and play \(\overline{\sigma}\) instead of \(\sigma\). The more substantive case arises when \(T^{\overline{\sigma} < \sigma}_{\overline{m}} \neq \emptyset\), as in Example~\ref{ex:undefeated} when comparing \(\overline{\sigma}\) to \(\sigma^{\text{Pooling}}\). In this case, there exists a type \(t'\) who sends \(\overline{m}\) in \(\overline{\sigma}\) but strictly prefers \(\sigma\) to \(\overline{\sigma}\), i.e., \(t' \in T^{\sigma < \overline{\sigma}}_{\overline{m}}\).

Unraveling occurs when any such type \(t'\), despite strictly preferring \(\sigma\) to \(\overline{\sigma}\), nonetheless has an incentive to deviate to \(\overline{m}\) and play \(\overline{\sigma}\) once certain ``other'' types of the sender have already deviated to \(\overline{m}\). These ``other'' types consist of (i) types who prefer \(\overline{\sigma}\) to \(\sigma\), i.e., \(T^{\overline{\sigma} \geq \sigma}_{\overline{m}}\), and (ii) types ranked higher than \(t'\), i.e., \(F^{\overline{\sigma} < \sigma}_{\overline{m}}\left(t'\right)\). If the sender of type \(t'\) adheres to \(\sigma\) and sends \(m'\), their payoff is at most equal to the equilibrium payoff in \(\overline{\sigma}\) if the receiver believes that only types other than those ``other'' types would still play \(\sigma\), forming a conditional belief \(\mu'\) derived from \(\mu\) by excluding those ``other'' types after observing \(m'\). As a result, they would like to deviate to \(\overline{\sigma}\). Unraveling proceeds sequentially from the highest-ranked type to the lowest-ranked type. The logic is similar to the classical analysis of voluntary disclosure in \citet{grossmanDisclosureLawsTakeover1980}, \citet{milgromGoodNewsBad1981} and \citet{verrecchiaDiscretionaryDisclosure1983}. Here, lower-ranked types voluntarily deviate to \(\overline{\sigma}\) once all higher-ranked types have already done so.

For instance, in Example~\ref{ex:undefeated}, the firm forms a belief such that \(\mu^{\text{Pooling}'}\left(\left.t_{H}\right|\tilde{m}'=0\right) = 0\) when the medium-type worker adheres to the pooling equilibrium. Under this belief, the medium-type worker's payoff (1.36) is lower than the equilibrium payoff in \(\overline{\sigma}\) (1.85). As a result, the medium-type worker would prefer to deviate to \(\overline{\sigma}\), even though they initially prefer \(\sigma^{\text{Pooling}}\) to \(\overline{\sigma}\). Hence, the interpretation of \(\overline{m}\) in \(\overline{\sigma}\) rather than in \(\sigma^{\text{Pooling}}\) triggers an unraveling.

Unraveling provides a consistent story for the receiver, explaining why all types who send \(\overline{m}\) in \(\overline{\sigma}\) would like to deviate to \(\overline{m}\) and play \(\overline{\sigma}\) instead of \(\sigma\) irrespective of their initial preferences for \(\overline{\sigma}\) or \(\sigma\). Hence, the receiver should interpret \(\overline{m}\) in \(\overline{\sigma}\) instead of in \(\sigma\), and we therefore say that \(\overline{\sigma}\) is more persuasive than \(\sigma\).

\begin{definition}\label{def:persuasive}
$\overline{\sigma}\in\text{SE}\left(G\right)$ is \emph{more persuasive} than $\sigma\in\text{SE}\left(G\right)$, if there exists a message \(\overline{m}\) on-path of \(\overline{\sigma}\) such that the interpretation of \(\overline{m}\) in \(\overline{\sigma}\) rather than in \(\sigma\) triggers an unraveling.
\end{definition}

Like the intuitive criterion and the D1 criterion, persuasiveness cannot distinguish between equilibria which are payoff-equivalent for the sender.\footnote{Equilibria which differ only in the receiver's off-path beliefs are payoff-equivalent for the sender. There might also exist rare cases where two equilibria differ in the sender's strategies but still generate the same payoff.} Hence, we define the most persuasive equilibrium as follows.

\begin{definition}\label{def:payoff-equivalence}
Two equilibria $\overline{\sigma}, \sigma\in\text{SE}\left(G\right)$ are \emph{payoff-equivalent} for the sender if \(u_{S}\left(t, \overline{\sigma}\right) = u_{S}\left(t, \sigma\right)\) for all \(t \in T\).
\end{definition}

\begin{definition}\label{def:most-persuasive}
$\overline{\sigma}\in\text{SE}\left(G\right)$ is \emph{most persuasive} if it is more persuasive than any other equilibrium $\sigma\in\text{SE}\left(G\right)$ that is not payoff-equivalent for the sender.
\end{definition}

If \(\overline{\sigma}\) is most persuasive, then for any other equilibrium \(\sigma\), we can alway find a message \(\overline{m}\) such that every type who sends \(\overline{m}\) in \(\overline{\sigma}\) would like to play \(\overline{\sigma}\) instead of \(\sigma\) because of unraveling. In Example~\ref{ex:undefeated}, we argue that the LMSE is more persuasive than both the Riley equilibrium and the pooling equilibrium. More generally, one can establish that the LMSE is the unique most persuasive equilibrium among all equilibria in monotone signaling games; we formalize this result in the next section. In some non-monotone signaling games, however, a most persuasive equilibrium may fail to exist or be unique (see Section~\ref{sec:limitations}).

\section{Monotone Signaling Games}\label{sec:monotone-signaling-games}

We study a class of monotone signaling games that satisfy the following assumptions, which includes \citet{spenceJobMarketSignaling1973}'s job market signaling model as an application.

\begin{assumption}\label{assumption:continuity-concavity}
    Continuity and Concavity:
    \begin{itemize}
        \item $T = \left\{1, 2, \dots, n\right\}$ is finite. 
        \item $M$ and $A$ are closed intervals of \(\mathbb{R}\).
        \item $u_{S}$ and $u_{R}$ are continuous in $m$ and $a$.
        \item $u_{R}$ is strictly concave in $a$.
    \end{itemize}
\end{assumption}

Assumption~\ref{assumption:continuity-concavity} ensures that \(\text{BR}\left(m, \mu\right)\), i.e., the receiver's best response correspondence after observing the message \(m\) under the belief \(\mu\), is always well-defined. In particular, it is a single-valued continuous function in \(m\).

\begin{assumption}\label{assumption:monotonicity}
    Monotonicity: 
    \begin{itemize}
        \item If \(a' > a\), then \(u_{S}\left(t,m,a'\right) > u_{S}\left(t,m,a\right)\) for all \(t\) and \(m\).
        \item \(\frac{\partial u_{R}}{\partial a}\) is a strictly increasing function of \(t\).
    \end{itemize}
\end{assumption}

Assumption~\ref{assumption:monotonicity} describes the sender's and the receiver's preferences. The sender prefers a higher action from the receiver. The receiver prefers to take a higher action when they believe the sender is of a higher type. In the job market signaling model, this means that the worker prefers a higher wage from the firm, and the firm prefers to offer a higher wage when they believe the worker is of a higher type (more productive).

\begin{assumption}\label{assumption:single-crossing}
    Single-Crossing:
    \begin{itemize}
        \item[] If \(m < m'\) and \(t < t'\), then \\
         \(u_{S}\left(t, m, a\right) \leq u_{S}\left(t, m', a'\right)\) implies that \(u_{S}\left(t', m, a\right) < u_{S}\left(t', m', a'\right)\).
    \end{itemize}
\end{assumption}

Assumption~\ref{assumption:single-crossing} is the Spence-Mirrlees single-crossing condition, which guarantees that the indifference curves of different types of the sender through a fixed message-action pair intersect only once. It captures the idea that higher messages are less costly for higher types to send than for lower types. In the job market signaling model, this means that the cost of acquiring higher education levels is decreasing in the worker's type.

For the next assumption, we introduce additional notation. For any non-empty subset \(K\) of \(T\), the \(K\)-conditional belief \(p_{K} \in \Delta\left(T\right)\) is defined as:
\[p_{K}\left(t\right) = \begin{cases}
    \frac{p\left(t\right)}{\sum_{t' \in K} p\left(t'\right)} & \text{if } t \in K,\\
    0 & \text{otherwise.}
\end{cases}\]
To simplify the notation, we write the receiver's best response to the message \(m\) under the \(K\)-conditional belief \(p_{K}\) as \(\text{BR}\left(m, K\right)\), i.e., \(\text{BR}\left(m, K\right) = \text{BR}\left(m, p_{K}\right)\).

\begin{assumption}\label{assumption:extreme-messages}
    Low-Cost and High-Cost Messages: Let \(m^{l} = \min\left\{m \in M\right\}\) and \(m^{h} = \max\left\{m \in M\right\}\).
    
    \begin{itemize}
        \item \(\forall t \in T\), \(\forall m \in M\), \(u_{S}\left(t, m^{l}, \text{BR}\left(m^{l}, \left\{1\right\}\right)\right) \geq u_{S}\left(t, m, \text{BR}\left(m, \left\{1\right\}\right)\right)\).
        \item \(\forall t \in T\setminus \left\{n\right\}\), \(u_{S}\left(t, m^{h}, \text{BR}\left(m^{h}, \left\{n\right\}\right)\right) < u_{S}\left(t, m^{l}, \text{BR}\left(m^{l}, \left\{1\right\}\right)\right).\)
    \end{itemize}
\end{assumption}

Assumption~\ref{assumption:extreme-messages} describes the message space. It states that the lowest message \(m^{l}\) is the cheapest message for all types of the sender to send, and the highest message \(m^{h}\) is the most expensive message such that no type of the sender, except possibly the highest type, would want to send it. In the job market signaling model, the worker incurs no cost for having zero education, while the level of education can go to infinity, which is too costly for every one.

\begin{definition}\label{def:monotone-signaling-game}
    A \emph{monotone signaling game} \(G_{\text{S}}\) is a signaling game that satisfies Assumptions~\ref{assumption:continuity-concavity}-\ref{assumption:extreme-messages}.
\end{definition}

Similar assumptions (including A\ref{assumption:message-concavity} below) have been applied in many general treatments of this class of games \citep*{rileyInformationalEquilibrium1979,choStrategicStabilityUniqueness1990,mailathBeliefBasedRefinementsSignalling1993}. Given that the message and action spaces are intervals, pure-strategy equilibria always exist. We denote the set of pure-strategy equilibria in the monotone signaling game \(G_{\text{S}}\) as \(\text{PSE}\left(G_{\text{S}}\right)\). Our analysis focuses on pure-strategy equilibria in \(G_{\text{S}}\), and we begin by formally defining the LMSE.

\begin{definition}\label{def:LMSE}
    \(\overline{\sigma} \in \text{PSE}\left(G_{\text{S}}\right)\) \emph{lexicographically dominates} (lex-dominates) \(\sigma \in \text{PSE}\left(G_{\text{S}}\right)\), if there exists \(\overline{t} \in T\) such that:
    \begin{itemize}
        \item \(u_{S}\left(\overline{t}, \overline{\sigma}\right) > u_{S}\left(\overline{t}, \sigma\right)\).
        \item \(u_{S}\left(t, \overline{\sigma}\right) \geq u_{S}\left(t, \sigma\right) \ \forall t > \overline{t}\).
    \end{itemize}
\end{definition}

\(\overline{\sigma} \in \text{PSE}\left(G_{\text{S}}\right)\) is a \emph{lexicographically maximum sequential equilibrium} (LMSE) if there exists no other equilibrium \(\sigma \in \text{PSE}\left(G_{\text{S}}\right)\) that lex-dominates \(\overline{\sigma}\). The equilibrium outcome of a LMSE is the \emph{lexicographically maximum} outcome (lex max outcome). A LMSE exists because pure-strategy equilibria exist in \(G_{\text{S}}\) \citep*{mailathBeliefBasedRefinementsSignalling1993}.

We now present our main results, which formalize the intuitions illustrated in the preceding examples and establish, step by step, that the lex max outcome is generically the unique most persuasive equilibrium outcome.

\begin{theorem}\label{thm:persuasive}
    In any game \(G_{\text{S}}\), the LMSE is most persuasive.
\end{theorem}

Proof Sketch: We need to show that the LMSE is more persuasive than any other equilibrium that is not payoff-equivalent for the sender. Given any other equilibrium \(\sigma\), we first identify the message \(\overline{m}\) on the equilibrium path of the LMSE that can be used to show that the interpretation of this message in the LMSE rather than in \(\sigma\) can trigger an unraveling. We pin down this message \(\overline{m}\) as the equilibrium message sent by the highest type who strictly prefers the LMSE to \(\sigma\). Next, we show that there exists a cutoff type among the set of types who send \(\overline{m}\) in the LMSE such that all types weakly above this cutoff type prefer the LMSE to \(\sigma\), while all types strictly below this cutoff type strictly prefer \(\sigma\) to the LMSE. To show the existence of an unraveling, we show that when the unraveling condition \eqref{eq:unraveling} is violated at some type, we can construct another equilibrium that lex-dominates the LMSE, which is a contradiction. See Appendix~\ref{thm:proof-persuasive} for the details of the proof.

Theorem~\ref{thm:persuasive} shows that for any alternative equilibrium that is not payoff-equivalent to the LMSE, the LMSE can always challenge this equilibrium with a new interpretation of some message. Under this new interpretation, all types who send this message in the LMSE would like to deviate to this message and play the LMSE instead of this alternative equilibrium because of unraveling. A natural question is whether an alternative equilibrium could similarly challenge the LMSE. The following result demonstrates that no such equilibrium exists.

\begin{theorem}\label{thm:utility-uniqueness}
    In any game \(G_{\text{S}}\), the most persuasive equilibrium is unique up to payoff equivalence for the sender.
\end{theorem}

Proof Sketch: We prove by contradiction. Suppose there exist two most persuasive equilibria, the LMSE and \(\hat{\sigma}\) that are not payoff-equivalent for the sender. Then, we can find a message \(\hat{m}\) on the equilibrium path of \(\hat{\sigma}\) such that the interpretation of this message in \(\hat{\sigma}\) rather than in the LMSE triggers an unraveling. In particular, there exists a type \(i\) who sends \(\hat{m}\) in \(\hat{\sigma}\) and strictly prefers \(\hat{\sigma}\) to the LMSE. However, there also exist types above \(i\) who send \(\hat{m}\) in \(\hat{\sigma}\) and strictly prefer the LMSE to \(\hat{\sigma}\), which means we do not have the cutoff structure we see in the proof of Theorem~\ref{thm:persuasive}. Then, we show that the unraveling condition \eqref{eq:unraveling} always fails when it comes to the highest-ranked type among the types identified above irrespective of the ranking function, which is a contradiction. See Appendix~\ref{thm:proof-utility-uniqueness} for the details of the proof.

Theorems~\ref{thm:utility-uniqueness} establishes that no alternative equilibrium can challenge the LMSE by providing a new interpretation of any message. Moreover, Theorems~\ref{thm:persuasive} and \ref{thm:utility-uniqueness} together imply that the interpretations of all messages in the game are determined by the LMSE, in the sense that no other equilibrium offers a more persuasive interpretation of any message.

We characterize the sender's equilibrium payoff in the most persuasive equilibrium. However, it does not guarantee a unique outcome. In order to ensure uniqueness, we assume the following as in \citet{choStrategicStabilityUniqueness1990} and \citet*{mailathBeliefBasedRefinementsSignalling1993}.

\begin{assumption}\label{assumption:message-concavity}
    \(u_{S}\left(t, m, \text{BR}\left(m, p\right)\right)\) is strictly quasi-concave in \(m\) for all \(t \in T\) and \(p \in \Delta\left(T\right)\).
\end{assumption}

\begin{theorem}\label{thm:outcome-uniqueness}
    Assume A5. Generically in the space of prior \(p \in \Delta\left(T\right)\), the lex max outcome is the unique most persuasive equilibrium outcome in the game \(G_{\text{S}}\).
\end{theorem}

The genericity of the result implies that the measure of the priors under which uniqueness fails is zero. See Appendix~\ref{thm:proof-outcome-uniqueness} for the details of the proof.

Under assumptions similar to A\ref{assumption:continuity-concavity}-A\ref{assumption:message-concavity}, \citet{choStrategicStabilityUniqueness1990} uniquely select the Riley outcome. When taking into account the Stiglitz critique, persuasiveness uniquely selects the lex max outcome. \citet{choStrategicStabilityUniqueness1990} introduce the following example to show that their result does not hold when the message space is discrete in monotone signaling games.

\begin{example}[Discrete Spencian Game]\label{ex:discrete}
    We follow the setup of Example~\ref{ex:stiglitz}. There are two types of a worker, \(T = \left\{t_{L}, t_{H}\right\}\), where \(t_{L} = \frac{2}{3}\) and \(t_{H} = 1\). The prior probabilities are \(p\left(t_{L}\right) = p\left(t_{H}\right) = 0.5\). The set of education levels is discrete, \(M = \left\{m_{0}, m_{1}\right\}\), where \(m_{0} = 0\) and \(m_{1} = 1\). The firm is the same as in Example~\ref{ex:stiglitz}. The payoff functions of workers are given by \(u_{S}\left(t_{L}, m, a\right) = a - \frac{m}{2}\) and \(u_{S}\left(t_{H}, m, a\right) = a - \frac{m}{4}\).

    There are essentially three equilibria in this game, and we summarize in Table~\ref{tab:discrete} below the equilibrium strategies and payoffs of the low-type and high-type workers in these equilibria, presented in the same format as Example~\ref{ex:stiglitz}.\footnote{The term ``essentially'' indicates that any other equilibria can differ only with respect to off-path beliefs.} Here, we allow mixed strategies. \(\sigma_{S}^{1}\left(t_{H}\right) = \frac{1}{3} m_{0} + \frac{2}{3} m_{1}\) implies that the high-type worker chooses \(m_{0}\) with probability \(\frac{1}{3}\) and \(m_{1}\) with probability \(\frac{2}{3}\).
    \begin{table}[ht!]
        \centering
        \bgroup
        \def\arraystretch{1.5}
        \begin{tabular}{ccc}
        \toprule
            & \(t_L\) & \(t_H\) \\
        \midrule
        \(\sigma_{S}^{\text{Riley}}\left(t\right)\)    &  \(m_{0}\) & \(m_{1}\)  \\
        \(\sigma_{S}^{1}\left(t\right)\)    & \(m_{0}\) & \(\frac{1}{3} m_{0} + \frac{2}{3} m_{1}\)  \\
        \(\sigma_{S}^{\text{Pooling}}\left(t\right)\)    & \(m_{0}\) & \(m_{0}\) \\
        \bottomrule
        \end{tabular}
        \quad
        \begin{tabular}{ccc}
        \toprule
            & \(t_L\) & \(t_H\) \\
        \midrule
        \(u_{S}\left(t,\sigma^{\text{Riley}}\right)\)    &  \(\frac{2}{3}\) & \(\frac{3}{4}\)  \\
        \(u_{S}\left(t,\sigma^{1}\right)\)    & \(\frac{3}{4}\) & \(\frac{3}{4}\)  \\
        \(u_{S}\left(t,\sigma^{\text{Pooling}}\right)\)    & \(\frac{5}{6}\) & \(\frac{5}{6}\)  \\
        \bottomrule
        \end{tabular}
        \egroup
        \caption{Equilibrium Strategies and Payoffs of the Worker in Example~\ref{ex:discrete}\label{tab:discrete}}
    \end{table}
\end{example}

All equilibria pass the D1 criterion (and the intuitive criterion), while the most persuasive equilibrium is the LMSE \(\sigma^{\text{Pooling}}\), which Pareto-dominates the other two equilibria. In this example, the high-type worker finds it too costly to separate from the low-type worker. However, it is not possible for the higher-type worker to signal their preference through an off-path message when applying the D1 criterion (or the intuitive criterion), because \(m_{0}\) is on-path of every equilibrium. In contrast, persuasiveness builds on how the receiver interprets messages in different equilibria, regardless of whether these messages (\(m_0\)) are on-path or off-path of either \(\sigma^{\text{Riley}}\) or \(\sigma^{1}\). In this example, the firm expects that the high-type worker would like to deviate to \(m_{0}\) and play \(\sigma^{\text{Pooling}}\) instead of any other equilibria because both types of the worker strictly benefit, which is captured by the fact that the LMSE \(\sigma^{\text{Pooling}}\) is most persuasive. \citet{choStrategicStabilityUniqueness1990} also introduce another example in their paper to show that their result does not hold when the action space is discrete, while we can still find a unique most persuasive equilibrium, i.e., the LMSE. Both examples suggest that persuasiveness has more selection power than the D1 criterion (and the intuitive criterion) when the message space or action space is discrete in monotone signaling games.\footnote{Extending Theorems~\ref{thm:persuasive} and \ref{thm:utility-uniqueness} to monotone signaling games with compact message and action spaces is non-trivial. We need to allow mixed-strategy equilibria to ensure the existence of equilibrium. Persuasiveness relies on the comparison of the payoffs of different types of the sender across different equilibria. The main challenge lies in bounding the payoff that a particular type of sender could obtain under certain adjustments to the belief, as specified in Definition~\ref{def:unraveling}.}

\section{Discussion}\label{sec:discussion}

So far, we have focused on monotone signaling games.\footnote{To the author's best knowledge, monotone signaling games are the only class of signaling games where the previous selection criteria have strong selection power. For general signaling games, although there exist equilibria that pass the intuitive criterion and the D1 criterion, the set of such equilibria is difficult to characterize. The set of equilibrium outcomes that pass these criteria is typically non-singleton and could be potentially large (see Example~\ref{ex:hiding} below). The applicability of these criteria is less clear in general signaling games. Hence, this paper focuses on monotone signaling games, and illustrates the applicability of persuasiveness to general signaling games by examples.} We now show the applicability of persuasiveness to more general signaling games by looking at examples that are discussed in the previous literature. In Section~\ref{sec:non-monotone}, we show that persuasiveness can have more selection power than other criteria in some non-monotone signaling games. In Section~\ref{sec:limitations}, we turn to other examples in which no existing criteria have strong selection power. These examples illustrate that cyclicality arises when there does not exist a unique most persuasive equilibrium outcome. We argue that persuasiveness offers insights into equilibrium selection. In such cases, it may be reasonable to consider the \emph{least} persuasive equilibrium. In Section~\ref{sec:cheap-talk}, we show how persuasiveness can be applied to cheap-talk games, where other criteria designed for signaling games have no selection power.

\subsection{Non-Monotone Signaling Games}\label{sec:non-monotone}

We first look at the famous Beer-Quiche example introduced by \citet{choSignalingGamesStable1987} to motivate the intuitive criterion.

\begin{example}[Beer-Quiche Game]\label{ex:beer-quiche}
    There are two types of the sender, \(T = \left\{t_w, t_s\right\}\), where \(t_w\) is a wimp type and \(t_s\) is a surly type. The prior probabilities are \(p\left(t_w\right) = 0.1\) and \(p\left(t_s\right) = 0.9\). The sender can choose to have either beer or quiche for breakfast, i.e., \(M = \left\{\text{beer}, \text{quiche}\right\}\). The receiver can choose to either challenge the sender to a duel or not, i.e., \(A = \left\{\text{duel}, \text{don't}\right\}\). The payoffs are given in Figure~\ref{fig:beer-quiche} below, where the first entry in each pair is the sender's payoff and the second entry is the receiver's payoff. For instance, if the sender is of type \(t_w\), has beer for breakfast, and the receiver chooses to duel, then the sender's payoff is 0 and the receiver's payoff is 1.
    \begin{figure}[ht!]
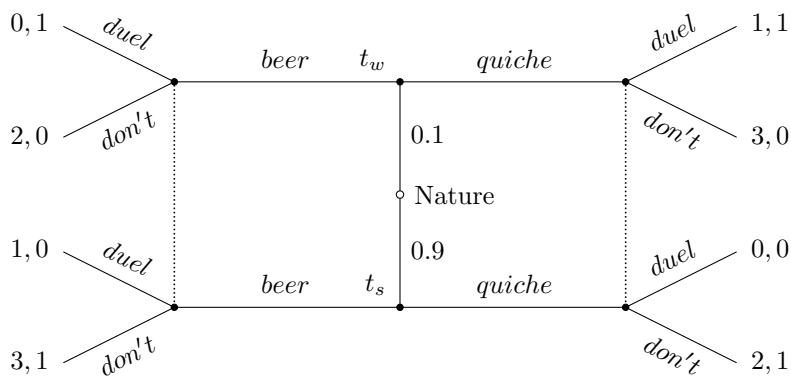

		\centering
		\begin{istgame}[font=\footnotesize]
			\xtdistance{15mm}{20mm}
			\istroot(0)[chance node]<0>{Nature}
			\istb<grow=north>{0.1}[r]
			\istb<grow=south>{0.9}[r]
			\endist
			\xtdistance{30mm}{15mm}
			\istroot(1)(0-1)<160>{\(t_w\)}
			\istb<grow=east>{quiche}[a]
			\istb<grow=west>{beer}[a]
			\endist
			\istroot(2)(0-2)<160>{\(t_s\)}
			\istb<grow=east>{quiche}[a]
			\istb<grow=west>{beer}[a]
			\endist
			\xtdistance{15mm}{15mm}
			\istroot'[east](q1)(1-1)
			\istb{duel}[above,sloped]{1,1}
			\istb{don't}[below,sloped]{3,0}
			\endist
			\istroot[west](b1)(1-2)
			\istb{duel}[above,sloped]{0,1}
			\istb{don't}[below,sloped]{2,0}
			\endist
			\istroot'[east](q2)(2-1)
			\istb{duel}[above,sloped]{0,0}
			\istb{don't}[below,sloped]{2,1}
			\endist
			\istroot[west](b2)(2-2)
			\istb{duel}[above,sloped]{1,0}
			\istb{don't}[below,sloped]{3,1}
			\endist
			\xtInfoset(q1)(q2)
			\xtInfoset(b1)(b2)
		\end{istgame}
        \caption{Beer-Quiche Game\label{fig:beer-quiche}}
	\end{figure}

    There are essentially two pooling equilibria \(\sigma^{\text{Beer}}\) and \(\sigma^{\text{Quiche}}\) in this game. In \(\sigma^{\text{Beer}}\), both types of the sender choose beer, and the receiver chooses to duel after observing quiche and not to duel after observing beer. In \(\sigma^{\text{Quiche}}\), both types of the sender choose quiche, and the receiver chooses not to duel after observing quiche and to duel after observing beer. Both the intuitive criterion and the D1 criterion selects \(\sigma^{\text{Beer}}\). It is easy to check that \(\sigma^{\text{Beer}}\) is more persuasive than \(\sigma^{\text{Quiche}}\), because the type \(t_{s}\) prefers \(\sigma^{\text{Beer}}\), and the type \(t_{w}\) would also like to choose beer because of unraveling.
\end{example}

In Example~\ref{ex:beer-quiche}, persuasiveness selects the same equilibrium outcome as both the intuitive criterion and the D1 criterion. However, this is not always the case in any non-monotone signaling game. In Example~\ref{ex:hiding} below, which is introduced by \citet{choSignalingGamesStable1987} to discuss their limitations, we show that persuasiveness can have more selection power than both the intuitive criterion and the D1 criterion.

\begin{example}[Hiding Game]\label{ex:hiding}
    There are two types of the sender, \(T = \left\{t_1, t_2\right\}\). The prior probabilities are \(p\left(t_1\right) = p\left(t_2\right) = 0.5\). The sender selects either message \(m_1\) or message \(m_2\), i.e., \(M = \left\{m_1, m_2\right\}\). The game terminates immediately following the choice of \(m_1\), without any subsequent action from the receiver. Only if \(m_2\) is chosen does the receiver have the opportunity to select an action from the set \(A = \left\{a_1, a_2, a_3\right\}\). The payoffs are given in Figure~\ref{fig:hiding} below, where the first entry in each pair is the sender's payoff and the second entry is the receiver's payoff. For instance, if the sender is of type \(t_1\), chooses \(m_2\), and the receiver chooses \(a_1\), then the sender's payoff is \(-1\) and the receiver's payoff is 3.
    \begin{figure}[ht!]
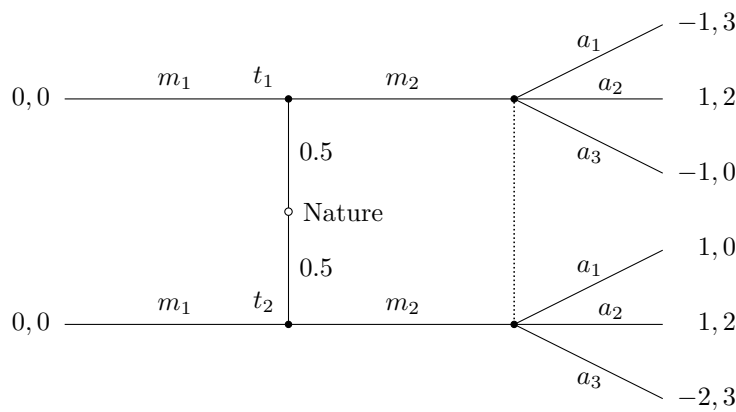

		\centering
		\begin{istgame}[font=\footnotesize]
			\xtdistance{15mm}{20mm}
			\istroot(0)[chance node]<0>{Nature}
			\istb<grow=north>{0.5}[r]
			\istb<grow=south>{0.5}[r]
			\endist
			\xtdistance{30mm}{15mm}
			\istroot(1)(0-1)<160>{\(t_1\)}
			\istb<grow=east>{m_2}[a]
			\istb<grow=west>{m_1}[a]{0,0}[l]
			\endist
			\istroot(2)(0-2)<160>{\(t_2\)}
			\istb<grow=east>{m_2}[a]
			\istb<grow=west>{m_1}[a]{0,0}[l]
			\endist
			\xtdistance{20mm}{10mm}
			\istroot'[east](q1)(1-1)
			\istb{a_1}[a]{-1,3}
			\istb{a_2}[ar]{\phantom{-}1,2}
			\istb{a_3}[b]{-1,0}
			\endist
			\istroot'[east](q2)(2-1)
			\istb{a_1}[a]{\phantom{-}1,0}
			\istb{a_2}[ar]{\phantom{-}1,2}
			\istb{a_3}[b]{-2,3}
			\endist
			\xtInfoset(q1)(q2)
		\end{istgame}
        \caption{Hiding Game\label{fig:hiding}}
	\end{figure}

    In this game, both types of the sender try to hide their types from the receiver, as we can see from the payoff structure in Figure~\ref{fig:hiding}. There are essentially two pooling equilibria \(\sigma^{m_1}\) and \(\sigma^{m_2}\) in this game. In \(\sigma^{m_1}\), both types of the sender choose \(m_1\), and the receiver chooses \(a_3\) after observing \(m_2\). In \(\sigma^{m_2}\), both types of the sender choose \(m_2\), and the receiver chooses \(a_2\) after observing \(m_2\). Notice that \(\sigma^{m_2}\) strictly Pareto-dominates \(\sigma^{m_1}\), as it yields higher payoffs for both the sender and the receiver. However, both equilibria pass the intuitive criterion and the D1 criterion. This is because both criteria emphasize how \emph{one} type of the sender would like to \emph{signal} to the receiver through an off-path message, but they do not account for situations in which \emph{two} types of the sender would like to \emph{hide} from the receiver only by \emph{jointly} sending an off-path message.\footnote{If we test \(\sigma^{m_1}\) by applying the two steps of the intuitive criterion as described at the start of Section~\ref{sec:persuasive}, then we get \(D = \left\{t_1, t_2\right\} \) in Step 1 because both types of the sender could benefit by deviating to \(m_2\). However, in Step 2, no type of the sender can profitably deviate to \(m_2\) under the least favorable belief of the receiver.} In contrast, persuasiveness selects \(\sigma^{m_2}\) as the most persuasive equilibrium, because both types of the sender would like to jointly deviate to \(m_2\) and play \(\sigma^{m_2}\) instead of \(\sigma^{m_1}\), which is captured by a trivial unraveling, i.e., the unraveling condition \eqref{eq:unraveling} is vacuously satisfied in this example. Note that persuasiveness not only considers the incentive of the sender to reveal their type as in monotone signaling games, but also takes into account the incentive of the sender to hide their types by pooling together as in this example.
\end{example}

In Table~\ref{tab:comparison} below, we summarize the selection results of some equilibrium refinements in the examples discussed so far. In Appendix~\ref{sec:refinements}, we offer an intuitive explanation for those criteria and show how they apply to the examples. Please refer to the original papers for the details of each selection criterion.
\begin{table}[ht!]
    \centering
    \caption{Comparison of Equilibrium Refinements in Examples\label{tab:comparison}}
    \vspace{0.2cm}
    \begin{tabular}{lcccc}
    \toprule
                            & Intuitive \& D1                & G-P     & Undefeated   & Persuasive \\
    \midrule
    Ex.~\ref{ex:stiglitz} (Two-Type) & \(\sigma^{\text{Riley}}\) & None    & LMSE          & LMSE      \\
    Ex.~\ref{ex:undefeated} (Three-Type)       & \(\sigma^{\text{Riley}}\) & None    & LMSE, Pooling         & LMSE       \\
    Ex.~\ref{ex:discrete} (Discrete)       & All & All    & All         & LMSE       \\
    Ex.~\ref{ex:beer-quiche} (Beer-Quiche)          & \(\sigma^{\text{Beer}}\)                      & \(\sigma^{\text{Beer}}\)    & \(\sigma^{\text{Beer}}, \sigma^{\text{Quiche}}\) & \(\sigma^{\text{Beer}}\)       \\
    Ex.~\ref{ex:hiding} (Hiding)            & \(\sigma^{m_1}, \sigma^{m_2}\)              & \(\sigma^{m_2}\) & \(\sigma^{m_2}\)      & \(\sigma^{m_2}\)    \\
    \bottomrule
    \end{tabular}
    \\
    \vspace{-0.2cm}
    \begin{footnotesize}
        \begin{enumerate}[leftmargin=2cm, rightmargin=2cm]
            \item The intuitive criterion cannot select \(\sigma^{\text{Riley}}\) in Ex.~\ref{ex:undefeated}, because there are more than two types of the sender. Otherwise, the intuitive criterion and the D1 criterion select the same equilibria in all other examples.
            \item G-P: Perfect Sequential Equilibrium \citep{grossmanPerfectSequentialEquilibrium1986}.
            \item Undefeated Equilibrium \citep*{mailathBeliefBasedRefinementsSignalling1993}.
        \end{enumerate}
    \end{footnotesize}
\end{table}

The table shows that persuasiveness uniquely selects the most persuasive outcome across all examples, while other criteria fail to yield a unique prediction in certain cases. It suggests that persuasiveness has strong selection power even beyond monotone signaling games. Note that all examples discussed so far have a unique most persuasive equilibrium outcome. However, this is not always the case in general signaling games, as we discuss next.

\subsection{Limitations}\label{sec:limitations}

In this section, we examine two examples in which no existing criteria have strong selection power. They also do not admit a unique most persuasive equilibrium outcome. The first example, attributed to Kreps, is introduced by \citet{grossmanPerfectSequentialEquilibrium1986} to highlight the coordination problem between the sender and the receiver. The second example is proposed by \citet*{mailathBeliefBasedRefinementsSignalling1993} to illustrate the difficulty of imposing plausible restrictions on off-path beliefs. When applying persuasiveness to both examples, the issues translate into the non-uniqueness and non-existence of the most persuasive equilibrium, respectively. However, we argue that persuasiveness is a useful concept in these examples by providing insights into equilibrium selection.

\begin{example}[Coordination Game]\label{ex:coordination}
    There are two types of the sender, \(T = \left\{t_1, t_2\right\}\). The prior probabilities are \(p\left(t_1\right) = p\left(t_2\right) = 0.5\). The sender selects either message \(m_1\) or message \(m_2\), i.e., \(M = \left\{m_1, m_2\right\}\). The receiver can choose to take either action \(a_1\) or action \(a_2\), i.e., \(A = \left\{a_1, a_2\right\}\). The payoffs are given in Figure~\ref{fig:coordination} below.
    \begin{figure}[ht!]
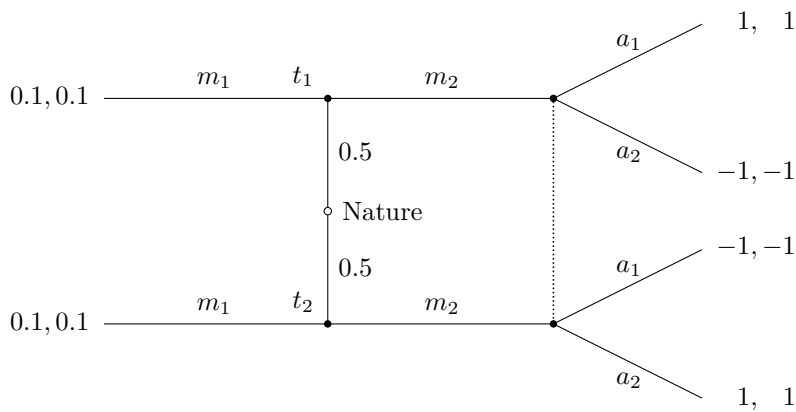

		\centering
		\begin{istgame}[font=\footnotesize]
			\xtdistance{15mm}{20mm}
			\istroot(0)[chance node]<0>{Nature}
			\istb<grow=north>{0.5}[r]
			\istb<grow=south>{0.5}[r]
			\endist
			\xtdistance{30mm}{15mm}
			\istroot(1)(0-1)<160>{\(t_1\)}
			\istb<grow=east>{m_2}[a]
			\istb<grow=west>{m_1}[a]{0.1,0.1}[l]
			\endist
			\istroot(2)(0-2)<160>{\(t_2\)}
			\istb<grow=east>{m_2}[a]
			\istb<grow=west>{m_1}[a]{0.1,0.1}[l]
			\endist
			\xtdistance{20mm}{20mm}
			\istroot'[east](q1)(1-1)
			\istb{a_1}[a]{\phantom{-}1,\phantom{-}1}
			\istb{a_2}[b]{-1,-1}
			\endist
			\istroot'[east](q2)(2-1)
			\istb{a_1}[a]{-1,-1}
			\istb{a_2}[b]{\phantom{-}1,\phantom{-}1}
			\endist
			\xtInfoset(q1)(q2)
		\end{istgame}
        \caption{Coordination Game\label{fig:coordination}}
	\end{figure}

    There are essentially three equilibria in this game, and we summarize in Table~\ref{tab:coordination} below the equilibrium strategies and payoffs of both types of the sender in these equilibria, presented in the same format as Example~\ref{ex:stiglitz}.
    \begin{table}[ht!]
        \centering
        \bgroup
        \def\arraystretch{1.5}
        \begin{tabular}{ccc}
        \toprule
            & \(t_1\) & \(t_2\) \\
        \midrule
        \(\sigma_{S}^{1}\left(t\right)\)    &  \(m_{2}\) & \(m_{1}\)  \\
        \(\sigma_{S}^{2}\left(t\right)\)    & \(m_{1}\) & \(m_{2}\)  \\
        \(\sigma_{S}^{\text{Pooling}}\left(t\right)\)    & \(m_{1}\) & \(m_{1}\) \\
        \bottomrule
        \end{tabular}
        \quad
        \begin{tabular}{ccc}
        \toprule
            & \(t_1\) & \(t_2\) \\
        \midrule
        \(u_{S}\left(t,\sigma^{1}\right)\)    &  \(1\) & \(0.1\)  \\
        \(u_{S}\left(t,\sigma^{2}\right)\)    & \(0.1\) & \(1\)  \\
        \(u_{S}\left(t,\sigma^{\text{Pooling}}\right)\)    & \(0.1\) & \(0.1\)  \\
        \bottomrule
        \end{tabular}
        \egroup
        \caption{Equilibrium Strategies and Payoffs of the Sender in Example~\ref{ex:coordination}\label{tab:coordination}}
    \end{table}

    The first two equilibria, \(\sigma^{1}\) and \(\sigma^{2}\), are separating equilibria in which the receiver takes different actions that best respond to different types of the sender in each equilibrium. The third equilibrium, \(\sigma^{\text{Pooling}}\), is a pooling equilibrium in which both types of the sender send \(m_1\), and the receiver randomizes between the two actions with equal probabilities after observing \(m_2\). All equilibria pass the intuitive criterion and the D1 criterion.\footnote{\(\sigma^{1}\) and \(\sigma^{2}\) are both undefeated equilibria and perfect sequential equilibria.}

    The coordination problem in this example arises because, a priori, there is no compelling reason for the sender and receiver to coordinate on either \(\sigma^{1}\) or \(\sigma^{2}\). Indeed, if coordination were achieved, both types of the sender would strictly prefer to send \(m_{2}\). Notice that if the sender could make a ``speech'' to the receiver, the ``speech'' would be used as a coordination device to select either \(\sigma^{1}\) or \(\sigma^{2}\), which creates a cheap-talk equilibrium in which both types of the sender obtain a payoff of 1 \citep{renyNaturalLanguageEquilibrium2025}. However, in the absence of such a ``speech,'' the sender and receiver cannot coordinate on either equilibrium.\footnote{It also shows that we cannot always rely on the implicit ``speech'' to interpret equilibrium selection criteria such as the intuitive criterion and the D1 criterion.} In this example, it becomes more reasonable to expect that the pooling equilibrium \(\sigma^{\text{Pooling}}\) would be played, because the receiver would like to randomize in order to avoid miscoordination after observing \(m_2\). Consequently, \(\sigma^{\text{Pooling}}\) can be interpreted as a ``safe'' equilibrium outcome.

    When applying persuasiveness to this example, the coordination problem is captured by the fact that \(\sigma^{1}\) is more persuasive than \(\sigma^{2}\), while at the same time \(\sigma^{2}\) is more persuasive than \(\sigma^{1}\). Although both equilibria more persuasive than the pooling equilibrium\(\sigma^{\text{Pooling}}\), they exhibit a cyclical relationship with one another. Hence, there does not exist a unique most persuasive equilibrium outcome. Non-uniqueness highlights the existence of multiple interpretations of \(m_2\) in either \(\sigma^{1}\) or \(\sigma^{2}\). In this example, the coordination problem translates into multiple interpretations of the same message due to the non-uniqueness of the most persuasive equilibrium. As we discussed above, the least persuasive equilibrium \(\sigma^{\text{Pooling}}\) is more appealing, which faces no issue of multiple interpretations of \(m_2\).
\end{example}

\begin{definition}\label{def:least-persuasive}
$\underline{\sigma}\in\text{SE}\left(G\right)$ is \emph{least persuasive} if any other equilibrium $\sigma\in\text{SE}\left(G\right)$, that is not payoff-equivalent for the sender, is more persuasive than $\underline{\sigma}$.
\end{definition}

The above example illustrates the coordination problem that arises when the interests of the sender and receiver are perfectly aligned. We now turn to a contrasting example in which their interests are misaligned. In this case, no most persuasive equilibrium exists. The interpretations of certain messages in some equilibria are ambiguous due to the presence of alternative, more persuasive equilibria. Consequently, upon observing such a message, the receiver may retain a reasonable doubt regarding the sender's type.

\begin{example}[Reasonable Doubt Game]\label{ex:doubt}
    There are three types of the sender, \(T = \left\{t_1, t_2, t_3\right\}\). The prior probabilities are \(p\left(t_1\right) = p\left(t_2\right) = p\left(t_3\right) = \frac{1}{3}\). The sender selects one message from \(M = \left\{m_1, m_2, m_3, m_4\right\}\). The game ends immediately after \(m_4\), otherwise the receiver can choose one action from \(A = \left\{a_1, a_2, a_3\right\}\). The payoffs are given in Figure~\ref{fig:doubt} below. There are essentially four equilibria in this game, and we summarize in Table~\ref{tab:doubt} below the equilibrium strategies and payoffs of the sender in all equilibria.
    \begin{figure}[ht!]
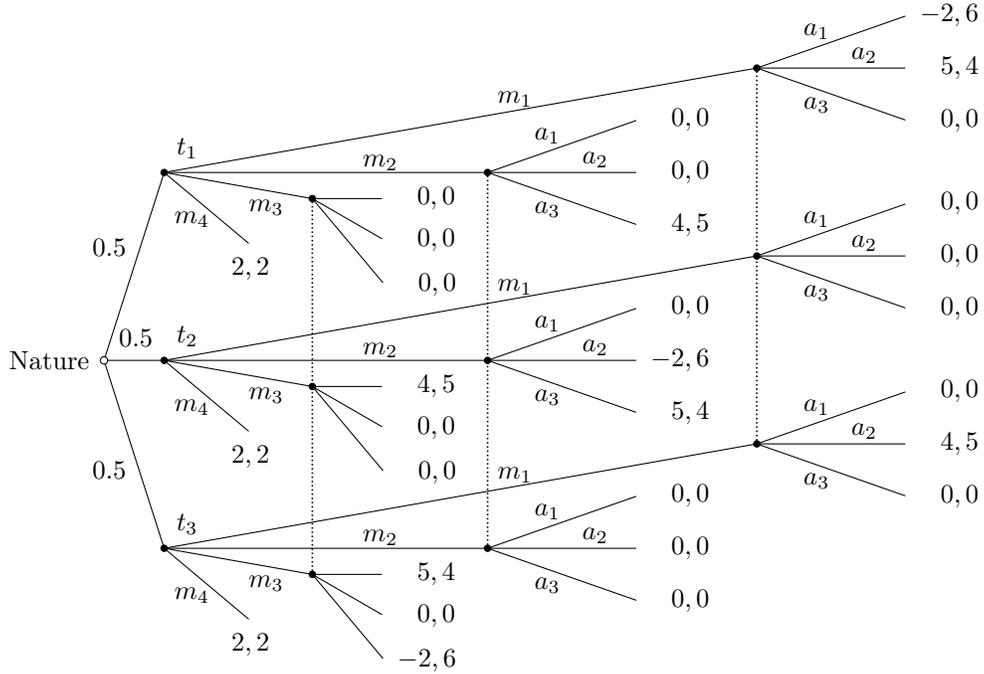

        \centering
        \begin{istgame}[font=\footnotesize]
            \setistgrowdirection'{east}
            \xtdistance{8mm}{25mm}
            \istroot(0)[chance node]<180>{Nature}
            \istb{0.5}[al]
            \istb{0.5}[a]
            \istb{0.5}[bl]
            \endist
            \istroot(1)(0-1)<60>{\(t_1\)}
            \istb<grow = 10, level distance=80mm>{m_1}[xshift=20pt, yshift=8pt]
            \istb<grow = 0, level distance=43mm>{m_2}[xshift=20pt, yshift=4pt]
            \istb<grow = -10, level distance=20mm>{m_3}[xshift=10pt, yshift=-8pt]
            \istb<grow = -40, level distance=15mm>{m_4}[xshift=-6pt, yshift=-4pt]{2, 2}[b]
            \endist
            \xtdistance{20mm}{7mm}
            \istroot(11)(1-1)
            \istbt{a_1}[xshift=-6pt, yshift=4pt]{-2, 6}
            \istbt{a_2}[xshift=12pt, yshift=5pt]{\phantom{-}5, 4}
            \istbt{a_3}[xshift=-6pt, yshift=-4pt]{\phantom{-}0, 0}
            \endist
            \istroot(12)(1-2)
            \istbt{a_1}[xshift=-6pt, yshift=4pt]{\phantom{-}0, 0}
            \istbt{a_2}[xshift=12pt, yshift=5pt]{\phantom{-}0, 0}
            \istbt{a_3}[xshift=-6pt, yshift=-4pt]{\phantom{-}4, 5}
            \endist
            \istroot(13)(1-3)
            \istbA<grow = 0, level distance=9.5mm>{}{\phantom{-}0, 0}
            \istbA<grow = -30, level distance=11mm>{}{\phantom{-}0, 0}
            \istbA<grow = -50, level distance=15mm>{}{\phantom{-}0, 0}
            \endist
            \istroot(2)(0-2)<60>{\(t_2\)}
            \istb<grow = 10, level distance=80mm>{m_1}[xshift=20pt, yshift=8pt]
            \istb<grow = 0, level distance=43mm>{m_2}[xshift=20pt, yshift=4pt]
            \istb<grow = -10, level distance=20mm>{m_3}[xshift=10pt, yshift=-8pt]
            \istb<grow = -40, level distance=15mm>{m_4}[xshift=-6pt, yshift=-4pt]{2, 2}[b]
            \endist
            \xtdistance{20mm}{7mm}
            \istroot(21)(2-1)
            \istbt{a_1}[xshift=-6pt, yshift=4pt]{\phantom{-}0, 0}
            \istbt{a_2}[xshift=12pt, yshift=5pt]{\phantom{-}0, 0}
            \istbt{a_3}[xshift=-6pt, yshift=-4pt]{\phantom{-}0, 0}
            \endist
            \istroot(22)(2-2)
            \istbt{a_1}[xshift=-6pt, yshift=4pt]{\phantom{-}0, 0}
            \istbt{a_2}[xshift=12pt, yshift=5pt]{-2, 6}
            \istbt{a_3}[xshift=-6pt, yshift=-4pt]{\phantom{-}5, 4}
            \endist
            \istroot(23)(2-3)
            \istbA<grow = 0, level distance=9.5mm>{}{\phantom{-}4,5}
            \istbA<grow = -30, level distance=11mm>{}{\phantom{-}0, 0}
            \istbA<grow = -50, level distance=15mm>{}{\phantom{-}0, 0}
            \endist
            \istroot(3)(0-3)<60>{\(t_3\)}
            \istb<grow = 10, level distance=80mm>{m_1}[xshift=20pt, yshift=8pt]
            \istb<grow = 0, level distance=43mm>{m_2}[xshift=20pt, yshift=4pt]
            \istb<grow = -10, level distance=20mm>{m_3}[xshift=10pt, yshift=-8pt]
            \istb<grow = -40, level distance=15mm>{m_4}[xshift=-6pt, yshift=-4pt]{2, 2}[b]
            \endist
            \xtdistance{20mm}{7mm}
            \istroot(31)(3-1)<-90>
            \istbt{a_1}[xshift=-6pt, yshift=4pt]{\phantom{-}0, 0}
            \istbt{a_2}[xshift=12pt, yshift=5pt]{\phantom{-}4,5}
            \istbt{a_3}[xshift=-6pt, yshift=-4pt]{\phantom{-}0, 0}
            \endist
            \istroot(32)(3-2)<-90>
            \istbt{a_1}[xshift=-6pt, yshift=4pt]{\phantom{-}0, 0}
            \istbt{a_2}[xshift=12pt, yshift=5pt]{\phantom{-}0, 0}
            \istbt{a_3}[xshift=-6pt, yshift=-4pt]{\phantom{-}0, 0}
            \endist
            \istroot(33)(3-3)<-90>
            \istbA<grow = 0, level distance=9.5mm>{}{\phantom{-}5, 4}
            \istbA<grow = -30, level distance=11mm>{}{\phantom{-}0, 0}
            \istbA<grow = -50, level distance=15mm>{}{-2, 6}
            \endist
            \xtInfoset(11)(31)
            \xtInfoset(12)(32)
            \xtInfoset(13)(33)
        \end{istgame}
        \caption{Reasonable Doubt Game\label{fig:doubt}}
    \end{figure}

    \begin{table}[ht!]
        \centering
        \bgroup
        \def\arraystretch{1.5}
        \begin{tabular}{cccc}
        \toprule
            & \(t_1\) & \(t_2\) & \(t_3\) \\
        \midrule
        \(\sigma_{S}^{1}\left(t\right)\)    &  \(m_1\) & \(m_4\) & \(m_1\) \\
        \(\sigma_{S}^{2}\left(t\right)\)    &  \(m_2\) & \(m_2\) & \(m_4\)\\
        \(\sigma_{S}^{3}\left(t\right)\)  &  \(m_4\) & \(m_3\) & \(m_3\) \\
        \(\sigma_{S}^{\text{Pooling}}\left(t\right)\)    & \(m_4\) & \(m_4\) & \(m_4\) \\
        \bottomrule
        \end{tabular}
        \quad
        \begin{tabular}{cccc}
        \toprule
            & \(t_1\) & \(t_2\) & \(t_3\) \\
        \midrule
        
        \(u_{S}\left(t,\sigma^{1}\right)\)    &  5 & 2 & 4  \\
        \(u_{S}\left(t,\sigma^{2}\right)\)    &  4 & 5 & 2  \\
        \(u_{S}\left(t,\sigma^{3}\right)\)  &  2 & 4 & 5  \\
        \(u_{S}\left(t,\sigma^{\text{Pooling}}\right)\)    & 2 & 2 & 2  \\
        \bottomrule
        \end{tabular}
        \egroup
        \caption{Equilibrium Strategies and Payoffs of the Sender in Example~\ref{ex:doubt}\label{tab:doubt}}
    \end{table}

    All equilibria pass the intuitive criterion and the D1 criterion.\footnote{No equilibrium is undefeated equilibrium or perfect sequential equilibrium.} Notice that the sender of type \(t_{i}\) attains the highest equilibrium payoff in \(\sigma^{i}\) for each \(i = 1, 2, 3\). However, for type \(t_i\) to obtain this payoff, type \(t_{i-1}\) must also be willing to send message \(m_i\); yet, in \(\sigma^{i-1}\), type \(t_{i-1}\) could achieve a higher equilibrium payoff by instead sending \(m_{i-1}\). Hence, when the receiver observes \(m_i\), they have a reasonable doubt that the sender might be of type \(t_{i}\), because only type \(t_{i}\) could achieve the highest equilibrium payoff by sending \(m_i\). This reasonable doubt could lead the receiver to take action \(a_i\) after observing \(m_i\) instead of action \(a_{i+1}\) in \(\sigma^{i}\). Then, type \(t_i\) would obtain a payoff of \(-2\) instead of 5 by sending \(m_i\).\footnote{When \(i = 4\), we relabel \(i = 1\). When \(i = 0\), we relabel \(i = 3\).} In this example, it becomes more reasonable to expect that the pooling equilibrium \(\sigma^{\text{Pooling}}\) would be played because of the existence of such reasonable doubts. Consequently, \(\sigma^{\text{Pooling}}\) can be interpreted as a ``safe'' equilibrium outcome.

    When applying persuasiveness to this example, those reasonable doubts are captured by the fact that \(\sigma^{i-1}\) is more persuasive than \(\sigma^{i}\), which leads to a cycle. They are all more persuasive than \(\sigma^{\text{Pooling}}\). Hence, there does not exist a most persuasive equilibrium. Non-existence casts doubt on the interpretation of \(m_i\) in the equilibrium \(\sigma^{i}\) because there exists another equilibrium \(\sigma^{i-1}\) that is more persuasive than \(\sigma^{i}\). As we discussed above, the least persuasive equilibrium \(\sigma^{\text{Pooling}}\) is more appealing, which is free from such reasonable doubts.
\end{example}

Neither the coordination game nor the reasonable doubt game admits a unique most persuasive equilibrium outcome. This non-uniqueness leads to multiple interpretations of the same message in different equilibria, while non-existence casts doubt on the interpretations of certain messages in some equilibria. In both examples, it appears more reasonable to expect that the least persuasive equilibrium would be played. This is not a coincidence. The term ``least persuasive'' can be somewhat misleading when a unique most persuasive equilibrium does not exist, as the argument that all other equilibria are more persuasive relies on the interpretations of messages that are subject to either multiple readings or reasonable doubts. Hence, the term ``least persuasive'' is a manifestation of the fact that this equilibrium is free from both issues and can be considered ``safe.''

In general, when a game does not admit a unique most persuasive equilibrium outcome, it is unclear which equilibrium will prevail. Nonetheless, the notion of persuasiveness remains useful, as it offers guidance for reasoning about equilibrium selection in such contexts. As demonstrated by the two examples, there are situations in which it is reasonable to focus on the least persuasive equilibrium.

\subsection{Cheap-Talk Games}\label{sec:cheap-talk}

Although we focus on (costly) signaling games in this paper, persuasiveness can also be applied to costless signaling games, i.e., cheap-talk games. Building on the setup of signaling games in Section~\ref{sec:setup}, cheap-talk games can be readily modeled by assuming that the set of messages \(M\) is infinite and that the payoff functions of both the sender and the receiver are independent of the message sent. We denote cheap-talk games as \(G_{\text{CT}}\). Note that the standard equilibrium refinements for signaling games, such as the intuitive criterion and the D1 criterion, have no selective power in cheap-talk games. This is because one can always support any equilibrium outcome with a ``noisy'' equilibrium in which all messages are sent with positive probability on the equilibrium path, so arguments that put plausible restrictions on off-path beliefs have no power to refine \citep[Section~3]{farrellMeaningCredibilityCheapTalk1993}. In contrast, persuasiveness emphasizes how the receiver interprets messages in equilibrium, which can select equilibria even when every message is on-path (see Example~\ref{ex:discrete} above).

We illustrate the selection power of persuasiveness in cheap-talk games with the following examples introduced in \citet{farrellMeaningCredibilityCheapTalk1993}.

\begin{example}[Cheap-Talk Games]\label{ex:cheap-talk}
    There are two types of the sender, \(T = \left\{t_1, t_2\right\}\). The prior probabilities are \(p\left(t_1\right) = p\) and \(p\left(t_2\right) = 1 - p\). The receiver has three different actions: \(a\left(t_1\right)\) and \(a\left(t_2\right)\) are best responses for the receiver when they are sufficiently confident that the sender is of type \(t_1\) or \(t_2\) respectively, and \(a\left(T\right)\) is the best response when the receiver has (close enough to) the prior probabilities in mind.\footnote{It is easy to construct payoff functions of the receiver that lead to Table~\ref{tab:cheap-talk}. We omit these constructions, as they are not essential to our analysis.} We consider the following three different cheap-talk games. The payoffs of the sender are given in Table~\ref{tab:cheap-talk} below.
    \begin{table}[ht!]
        \centering
        \bgroup
        \def\arraystretch{1.5}
        \begin{tabular}{ccc}
        \toprule
            & \(t_1\) & \(t_2\) \\
        \midrule
        \(a\left(t_1\right)\)    &  3 & 0  \\
        \(a\left(t_2\right)\)    &  0 & 3 \\
        \(a\left(T\right)\)    &  2 & 2  \\
        \bottomrule
        \multicolumn{3}{c}{\(G^{1}_{\text{CT}}\): I Will Tell You}
        \end{tabular}
        \quad
        \begin{tabular}{ccc}
        \toprule
            & \(t_1\) & \(t_2\) \\
        \midrule
        \(a\left(t_1\right)\)    &  1 & 0  \\
        \(a\left(t_2\right)\)    &  0 & 1 \\
        \(a\left(T\right)\)    &  2 & 2  \\
        \bottomrule
        \multicolumn{3}{c}{\(G^{2}_{\text{CT}}\): I Won't Tell You}
        \end{tabular}
        \quad
        \begin{tabular}{ccc}
        \toprule
            & \(t_1\) & \(t_2\) \\
        \midrule
        \(a\left(t_1\right)\)    &  2 & 1  \\
        \(a\left(t_2\right)\)    &  \(-1\) & 0 \\
        \(a\left(T\right)\)    &  0 & 2  \\
        \bottomrule
        \multicolumn{3}{c}{\(G^{3}_{\text{CT}}\): I Can't Tell You}
        \end{tabular}
        \egroup
        \caption{Payoffs of the Sender in Example~\ref{ex:cheap-talk}\label{tab:cheap-talk}}
    \end{table}

    In both \(G^{1}_{\text{CT}}\) and \(G^{2}_{\text{CT}}\), there are two equilibria: a babbling equilibrium where no information is transmitted and an informative equilibrium where the sender perfectly reveals their type. In \(G^{3}_{\text{CT}}\), there is a unique babbling equilibrium.

    We call \(G^{1}_{\text{CT}}\) the ``I Will Tell You'' game, because it is in both the sender's and receiver's interests to coordinate as in Example~\ref{ex:coordination}. Hence, the informative equilibrium is more appealing in \(G^{1}_{\text{CT}}\). In contrast, we call \(G^{2}_{\text{CT}}\) the ``I Won't Tell You'' game, because it is against both types of the sender's interests to reveal their types. Hence, the babbling equilibrium is more appealing in \(G^{2}_{\text{CT}}\). In both games, the more appealing equilibrium is the most persuasive equilibrium. In \(G^{3}_{\text{CT}}\), the babbling equilibrium is the unique equilibrium, which is vacuously the most persuasive equilibrium. We call it the ``I Can't Tell You'' game, because even if the sender of type \(t_1\) wants to tell the receiver that they are of type \(t_1\), they are unable to do so credibly as the sender of type \(t_2\) always imitates type \(t_1\).
\end{example}

Persuasiveness uniquely selects the most persuasive equilibrium in the above three cheap-talk games. It is related to \citet{farrellMeaningCredibilityCheapTalk1993}'s notion of neologism-proofness. However, this criterion is too demanding that the unique babbling equilibrium in \(G^{3}_{\text{CT}}\) is not neologism-proof. It captures the incentive of the sender of type \(t_1\) to reveal their type, but ignores the incentive of the sender of type \(t_2\) to mimic type \(t_1\) in an equilibrium---it does not account for whether the sender of type \(t_1\) can credibly reveal their type in an equilibrium.\footnote{It can also be seen as the Stiglitz critique in cheap-talk games \citep[p.~163]{rabinCommunicationRationalAgents1990}.} In contrast, persuasiveness considers both incentives and admits the unique babbling equilibrium in \(G^{3}_{\text{CT}}\).

Unfortunately, persuasiveness cannot uniquely select the most informative equilibrium in the Crawford-Sobel model \citep{crawfordStrategicInformationTransmission1982}\@. \citet*{chenSelectingCheapTalkEquilibria2008} propose an alternative criterion, NITS, which can uniquely select the most informative equilibrium under certain conditions. However, it is not clear how to extend their criterion to the current example or to more general cheap-talk games.\footnote{NITS is specifically designed for the Crawford-Sobel model. It requires a general notion of lowest type. It is not clear how to define the lowest type in general cheap-talk games such as the current example.}

\section{Related Literature}\label{sec:related-literature}

The literature on equilibrium selection is mostly based on the logic of forward induction. There are two main approaches. One is the axiomatic approach that begins with strategic stability initiated by \citet{kohlbergStrategicStabilityEquilibria1986}, which looks for equilibria that satisfy a list of axiomatic desiderata. Subsequent research seeks to refine and redefine the concept of stability \citep*{mertensStableEquilibriaReformulation1989,mertensStableEquilibriaReformulation1991, vandammeStableEquilibriaForward1989,hillasDefinitionStrategicStability1990,dilmeSequentiallyStableOutcomes2024}. Two recent papers that examine forward induction from a decision-theoretic and axiomatic approach are \citet{govindanForwardInduction2009,govindanAxiomaticEquilibriumSelection2012}. The other approach is the belief-based refinement, which is more directly motivated by putting plausible restrictions on off-path beliefs.\footnote{The two approaches are connected. Both the intuitive criterion and the divinity criterion are weaker than NWBR (Never a Weak Best Response), which is implied by strategic stability.} The two most well-known members of this family is the intuition criterion \citep{choSignalingGamesStable1987} and the divinity criterion \citep{banksEquilibriumSelectionSignaling1987}\@. \citet{choRefinementSequentialEquilibrium1987} extends the intuitive criterion to general extensive-form games\@. \citet{grossmanPerfectSequentialEquilibrium1986} refine the intuitive criterion. However, the G-P criterion suffers the non-existence problem in the job market signaling model \citep{spenceJobMarketSignaling1973}. Recent studies formalize the implicit ``speech'' used by \citet{choSignalingGamesStable1987} to motivate the intuitive criterion by adding cheap talk to signaling games \citep{clarkJustifiedCommunicationEquilibrium2021,renyNaturalLanguageEquilibrium2025}. Adding cheap talk to signaling games can refine equilibria; however, it may also select an equilibrium that does not exist in signaling games without cheap talk, as illustrated in Example~\ref{ex:coordination}.\footnote{For the equilibrium refinements in cheap-talk games, see \citet*{farrellMeaningCredibilityCheapTalk1993, rabinCommunicationRationalAgents1990, matthewsRefiningCheaptalkEquilibria1991,matthewsModelingCheapTalk1994,zapaterCredibleProposalsCommunication1997, chenSelectingCheapTalkEquilibria2008, semiratForwardneologismproofEquilibriumBetter2023,gordonEffectiveCommunicationCheapTalk2024}.}

There is also a strand of literature that studies equilibrium selection in a dynamic model from an evolutionary and learning perspective.\footnote{See also \citet{rabinDeviationsDynamicsEquilibrium1996}, \citet{umbhauerForwardInductionEvolutionary1997} and \citet{clarkJustifiedCommunicationEquilibrium2021}.} Building on the work of \citet*{kandoriLearningMutationLong1993} and \citet{youngEvolutionConventions1993}, \citet{noldekeDynamicModelEquilibrium1997} show that in a two-type Spencian game, the lex max outcome, which is uniquely selected by persuasiveness, is always contained in the unique recurrent set, while the Riley outcome may not be.\footnote{They use the term ``Hellwig equilibrium'' to refer to the best pooling equilibrium for the high-type sender \citep{hellwigRecentDevelopmentsTheory1987}. When the Riley equilibrium lex-dominates the Hellwig equilibrium, the Riley outcome is the unique outcome in the recurrent set. When the Hellwig equilibrium lex-dominates the Riley equilibrium, the Hellwig outcome is always contained in the unique recurrent set, while the Riley outcome  may not be. In both cases, the lex outcome is always selected.}

The Stiglitz critique \citep[p.~203]{choSignalingGamesStable1987} calls into question the general sorts of arguments used in these criteria to refine equilibria \citep{mailathReformulationCriticismIntuitive1988}, when they uniquely select the Riley outcome  in Example~\ref{ex:stiglitz}. In response to that, \citet*{mailathBeliefBasedRefinementsSignalling1993} propose the notion of undefeated equilibrium.\footnote{\citet{umbhauerForwardInductionConsistency1991} develops the consistent forward induction criterion which is similar to the undefeated equilibrium.} Persuasiveness is related to, yet distinct from, the undefeated equilibrium. Specifically, we do not require that every type of the sender who sends a message in a more persuasive equilibrium must prefer this equilibrium to the putative equilibrium (see Appendix~\ref{sec:refinements} for details).\footnote{We also do not require this message in the more persuasive equilibrium to be off-path of the putative equilibrium.} The unraveling condition \eqref{eq:unraveling} ensures that the set of types who initially strictly prefer the putative equilibrium to the more persuasive one would eventually deviate. This distinction is crucial for establishing the uniqueness of the most persuasive equilibrium in monotone signaling games (Theorem~\ref{thm:utility-uniqueness}), given that the undefeated equilibrium is typically not unique (see Table~\ref{tab:comparison} above).

The Stiglitz critique is similar to the motivation of Wilson equilibrium \citep{wilsonModelInsuranceMarkets1977} in insurance markets with adverse selection \citep{rothschildEquilibriumCompetitiveInsurance1976}\@. \citet[Section~3]{hellwigRecentDevelopmentsTheory1987} poses the question of why different equilibrium outcomes are selected in different contexts depending on which side---the informed or uninformed player---moves first, even though the underlying intuitive principles guiding these selections appear similar. Specifically, the Riley outcome is selected (by the intuitive criterion) when the customer proposes contracts, whereas the Wilson outcome arises when the insurance company does. This paper addresses the Stiglitz critique and provides an answer to this question. It shows that persuasiveness selects the same outcome in a signaling game where the customer moves first as the Wilson outcome in a screening game where the insurance company moves first, since both correspond to the lex max outcome.\footnote{To ensure the consistency between the signaling and screening games, we rule out cross-subsidization \citep{miyazakiRatRaceInternal1977}, because the customer cannot propose a menus of contracts in a signaling game. For a game theoretical foundation of the equilibrium selection in insurance markets with cross-subsidization, see \citet{mimraGameTheoreticFoundationWilson2011} and \citet{netzerGameTheoreticFoundation2014}.} Hence, the discrepancy in equilibrium selection in different contexts disappears, when we apply persuasiveness to signaling games.

The unraveling logic behind persuasiveness is related to the literature on information disclosure. See \citet{grossmanDisclosureLawsTakeover1980}, \citet{grossmanInformationalRoleWarranties1981}, \citet{milgromGoodNewsBad1981}, \citet{verrecchiaDiscretionaryDisclosure1983}, and \citet{madaraszCostContentInformation2025}. The iterative elimination of types of the sender from playing the putative equilibrium connects to the literature on the iterative elimination of strictly dominated strategies \citep*{abreuVirtualImplementationIteratively1992, kaponUsingDivideandConquerImprove2024}.\footnote{It is also related to the literature on the use of divide-and-conquer mechanisms to implement desirable social outcomes under all rationalizable strategy profiles. See also \citet*{segalNakedExclusionComment2000,spieglerExtractingInteractionCreatedSurplus2000, segalCoordinationDiscriminationContracting2003,winterIncentivesDiscrimination2004, boBribingVoters2007,eliazXgames2015,halacRaisingCapitalHeterogeneous2020,halacRankUncertaintyOrganizations2021}.}

\section{Concluding Remarks}\label{sec:conclusion}

This paper introduces a novel criterion, called persuasiveness, to select equilibria in signaling games. Persuasiveness is immune to the Stiglitz critique. It builds on how the receiver interprets messages in different equilibria, regardless of whether these messages are on-path or off-path. An equilibrium is more persuasive than an alternative one if there exists a message on its equilibrium path such that every type of the sender who sends that message in the equilibrium would like to deviate from the alternative equilibrium to that message because of unraveling. An equilibrium is most persuasive if it is more persuasive than any other equilibrium that is not payoff-equivalent for the sender. When there exists a unique most persuasive equilibrium, the interpretations of all messages in the game are determined, in the sense that no other equilibrium can provide a more persuasive interpretation of any message.

Persuasiveness has strong selective power. In monotone signaling games, it uniquely selects the most persuasive equilibrium outcome---the lex max outcome. In some non-monotone signaling games, it has stronger selective power than other existing equilibrium refinements. Furthermore, persuasiveness can have good selective power in cheap-talk games, where standard equilibrium refinements for signaling games have no selective power.

We conclude by posing three questions for future research. First, when a unique most persuasive equilibrium does not exist, how should equilibrium selection be approached? As illustrated by the two examples in Section~\ref{sec:limitations}, it may be plausible to expect that the least persuasive equilibrium prevails. Whether this observation generalizes to a broader class of signaling games, however, remains unclear. Second, can the concept of persuasiveness be extended to more general multi-stage games with multiple players? Intuitively, such an extension would require applying persuasiveness at each equilibrium history rather than to the equilibrium as a whole. It remains unclear whether this extension can be formulated in a concise manner and how it would relate to existing refinements for general extensive-form games \citep{choRefinementSequentialEquilibrium1987, govindanForwardInduction2009}. Third, can we identify a criterion that uniquely selects an equilibrium outcome in both the Spencian signaling games and Crawford-Sobel cheap-talk games, the two canonical models of communication under asymmetric information? Existing selection criteria for signaling games typically do not apply to cheap-talk games, and vice versa. Persuasiveness suggests the potential for a unified criterion, as it has selective power in both frameworks; nevertheless, it fails to uniquely select the most informative equilibrium in \citet{crawfordStrategicInformationTransmission1982}. Whether a unified criterion can be established that uniquely selects an equilibrium outcome in both classes of games remains an open question.

\bibliographystyle{econ}
\bibliography{PS-Zeng}

\appendix

\section{Proofs Omitted from the Main Text}\label{sec:proofs}
We first introduce some lemmas that will be used in the proofs of main theorems later.

\begin{lemma}\label{lemma:monotone-message}
    Let \(\sigma \in \text{PSE}\left(G_{\text{S}}\right)\). If \(t < t'\), then \(\sigma_{S}\left(t\right) \leq \sigma_{S}\left(t'\right)\).
\end{lemma}
\begin{proof}
    Suppose not. Let \(m = \sigma_{S}\left(t\right)\) and \(m' = \sigma_{S}\left(t'\right)\). Let \(a = \sigma_{R}\left(m\right)\) and \(a' = \sigma_{R}\left(m'\right)\). By the equilibrium condition of \(\sigma\), \(u_{S}\left(t, m, a\right) \geq u_{S}\left(t, m', a'\right)\) and \(u_{S}\left(t', m', a'\right) \geq u_{S}\left(t', m, a\right)\). By the single-crossing condition (A\ref{assumption:single-crossing}), \(m > m'\) implies that \(u_{S}\left(t', m, a\right) > u_{S}\left(t', m', a'\right)\), which is a contradiction.
\end{proof}

\begin{lemma}\label{lemma:reverse-single-crossing}
    Reverse Single-Crossing:
    \begin{itemize}
        \item[] If \(m < m'\) and \(t < t'\), then \\
         \(u_{S}\left(t', m, a\right) \geq u_{S}\left(t', m', a'\right)\) implies that \(u_{S}\left(t, m, a\right) > u_{S}\left(t, m', a'\right)\).
    \end{itemize}
\end{lemma}
\begin{proof}
    Suppose not. By the single-crossing condition (A\ref{assumption:single-crossing}), \(u_{S}\left(t, m, a\right) \leq u_{S}\left(t, m', a'\right)\) implies that \(u_{S}\left(t', m, a\right) < u_{S}\left(t', m', a'\right)\), which is a contradiction.
\end{proof}

For the next lemma, we introduce an additional notation. We study the game truncated from \(G_{\text{S}}\) by restricting the sender's types to be a subset of the original set. Let
    \[
    T^{j} = \left\{1, 2, \dots, j \right\} \quad \text{and} \quad p^j\left(t\right) = p_{T^{j}}\left(t\right).
    \]
The truncated game \(G_{\text{S}}^{j}\) is defined as the original \(G_{\text{S}}\) except that the sender's type space is \(T^{j}\) and the prior distribution \(p^{j}\) is the \(T^{j}\)-conditional belief \(p_{T^{j}}\). Given any equilibrium \(\sigma \in \text{PSE}\left(G_{\text{S}}\right)\), we can construct a \(j\)-truncated strategy \(\sigma^{j}\) of the truncated game \(G_{\text{S}}^{j}\) by simply deleting those types higher than \(j\). As long as no type higher than \(j\) is pooling with \(j\) in the equilibrium \(\sigma\), we have a \(j\)-truncated equilibrium \(\sigma^{j} \in \text{PSE}\left(G_{\text{S}}^{j}\right)\).

\begin{corollary*}[\citet*{mailathBeliefBasedRefinementsSignalling1993}]\label{lemma:mailath}
    Suppose \(\tilde{\sigma} \in \text{PSE}\left(G_{\text{S}}\right)\) and \(\hat{\sigma}^{j} \in \text{PSE}\left(G_{\text{S}}^{j}\right)\) for some \(j < n\). Suppose further that \(u_{S}\left(j, \hat{\sigma}^{j}\right) \geq u_{S}\left(j, \tilde{\sigma}\right)\). Then, there exists \(\sigma \in \text{PSE}\left(G_{\text{S}}\right)\) such that:
    \begin{align*}
        u_{S}\left(t, \sigma\right) &\geq u_{S}\left(t, \hat{\sigma}^{j}\right)  \quad \forall t \leq j,\\
        u_{S}\left(t, \sigma\right) &\geq u_{S}\left(t, \tilde{\sigma}\right)  \quad\forall t > j.
    \end{align*}
\end{corollary*}

\begin{lemma}\label{lemma:LMSE}
    Let \(\overline{\sigma} \in \text{PSE}\left(G_{\text{S}}\right)\) be the LMSE. Then, for any \(j \in T\), and any equilibrium \(\sigma^{j} \in \text{PSE}\left(G_{\text{S}}^{j}\right)\), we have \(u_{S}\left(j, \overline{\sigma}\right) \geq u_{S}\left(j, \sigma^{j}\right)\).
\end{lemma}
\begin{proof}
    Suppose not. If \(j = n\), then \(u_{S}\left(n, \overline{\sigma}\right) < u_{S}\left(n, \sigma^{n}\right)\), which contradicts the definition of \(\overline{\sigma}\) as the LMSE. If \(j < n\), we have \(u_{S}\left(j, \overline{\sigma}\right) < u_{S}\left(j, \sigma^{j}\right)\). Then, by the previous corollary, there exists a new equilibrium \(\tilde{\sigma} \in \text{PSE}\left(G_{\text{S}}\right)\) such that
    \begin{align*}
        u_{S}\left(t, \tilde{\sigma}\right) &\geq u_{S}\left(t, \sigma^{j}\right)  \quad \forall t \leq j,\\
        u_{S}\left(t, \tilde{\sigma}\right) &\geq u_{S}\left(t, \overline{\sigma}\right)  \quad \forall t > j.
    \end{align*}
    In particular, \(u_{S}\left(j, \tilde{\sigma}\right) \geq u_{S}\left(j, \sigma^{j}\right) > u_{S}\left(j, \overline{\sigma}\right)\). Hence, \(\tilde{\sigma}\) lex-dominates \(\overline{\sigma}\), which contradicts the definition of \(\overline{\sigma}\) as the LMSE.
\end{proof}

\subsection{Proof of Theorem~\ref{thm:persuasive}}\label{thm:proof-persuasive}
We denote the LMSE by \(\overline{\sigma} \in \text{PSE}\left(G_{\text{S}}\right)\). Consider any other equilibrium \(\sigma \in \text{PSE}\left(G_{\text{S}}\right)\), in which \(u_{S}\left(t, \sigma\right) \neq u_{S}\left(t, \overline{\sigma}\right)\) for some \(t \in T\). To show that \(\overline{\sigma}\) is most persuasive, we need to prove that \(\overline{\sigma}\) is more persuasive than \(\sigma\). We first identify the message \(\overline{m}\), whose interpretation in \(\overline{\sigma}\) rather than in \(\sigma\) can trigger an unraveling. We define \(\overline{m}\) as follows:
\begin{align*}
    \overline{t} &= \max\left\{ \left.t \in T\right|u_{S}\left(t, \overline{\sigma}\right)>u_{S}\left(t, \sigma\right)\right\}, \\
    \overline{m} &= \overline{\sigma}_{S}\left(\overline{t}\right), \\
    T^{\overline{\sigma}}_{\overline{m}} &= \left\{ \left.t \in T\right| \overline{\sigma}_{S}\left(t\right)= \overline{m} \right\}.
\end{align*}

We start by showing that there exists a cutoff type \(t^* \in T^{\overline{\sigma}}_{\overline{m}}\) such that for any \(t \in T^{\overline{\sigma}}_{\overline{m}}\) and \(t \geq t^*\), we have \(t \in T^{\overline{\sigma}\geq \sigma}_{\overline{m}}\), while for any \(t \in T^{\overline{\sigma}}_{\overline{m}}\) and \(t < t^*\), we have \(t \in T^{\overline{\sigma} < \sigma}_{\overline{m}}\). By the definition of \(\overline{\sigma}\) and \(\overline{t}\), we have \(u_{S}\left(t,\overline{\sigma}\right) = u_{S}\left(t,\sigma\right)\) for all \(t > \overline{t}\). Hence, if \(t > \overline{t}\) and \(t \in T^{\overline{\sigma}}_{\overline{m}}\), we have \(t \in T^{\overline{\sigma}\geq \sigma}_{\overline{m}}\). Suppose now that \(t' < \overline{t}\) and \(t' \in T^{\overline{\sigma}\geq \sigma}_{\overline{m}}\), we claim that for any \(t'' > t'\) and \(t'' < \overline{t}\), we have \(t''\) strictly prefers \(\overline{\sigma}\) to \(\sigma\). Otherwise, we have
\begin{align}
    u_{S}\left(t', \overline{m}, \overline{\sigma}_{R}\left(\overline{m}\right)\right) &= u_{S}\left(t', \overline{\sigma}\right) &&\geq & u_{S}\left(t', \sigma\right) &= u_{S}\left(t',m', \sigma_{R}\left(m'\right)\right) \label{eq:se-low}\\
    u_{S}\left(t'', \overline{m}, \overline{\sigma}_{R}\left(\overline{m}\right)\right) &= u_{S}\left(t'', \overline{\sigma}\right) && \leq  & u_{S}\left(t'', \sigma\right) &= u_{S}\left(t'', m'', \sigma_{R}\left(m''\right)\right) \label{eq:se-medium}\\
    u_{S}\left(\overline{t}, \overline{m}, \overline{\sigma}_{R}\left(\overline{m}\right)\right) &= u_{S}\left(\overline{t}, \overline{\sigma}\right) && > & u_{S}\left(\overline{t}, \sigma\right) &= u_{S}\left(\overline{t}, m, \sigma_{R}\left(m\right)\right) \label{eq:se}
\end{align}
First, \(\overline{m} \geq m''\), otherwise the single-crossing condition (A\ref{assumption:single-crossing}) and \eqref{eq:se-medium} imply that
\[
u_{S}\left(\overline{t}, \overline{m}, \overline{\sigma}_{R}\left(\overline{m}\right)\right) < u_{S}\left(\overline{t}, m'', \sigma_{R}\left(m''\right)\right) \leq u_{S}\left(\overline{t}, m, \sigma_{R}\left(m\right)\right),
\]
which contradicts \eqref{eq:se}. Second, \(\overline{m} \leq m''\), otherwise the reverse single-crossing condition (Lemma~\ref{lemma:reverse-single-crossing}) and \eqref{eq:se-medium} imply that
\[
u_{S}\left(t', \overline{m}, \overline{\sigma}_{R}\left(\overline{m}\right)\right) < u_{S}\left(t', m'', \sigma_{R}\left(m''\right)\right) \leq u_{S}\left(t',m', \sigma_{R}\left(m'\right)\right),
\]
which contradicts \eqref{eq:se-low}. Then, we have \(\overline{m} = m''\). By monotonicity (A\ref{assumption:monotonicity}), \eqref{eq:se-medium} implies that
\[
u_{S}\left(\overline{t}, \overline{m}, \overline{\sigma}_{R}\left(\overline{m}\right)\right) \leq u_{S}\left(\overline{t}, \overline{m}, \sigma_{R}\left(\overline{m}\right)\right) \leq u_{S}\left(\overline{t}, m, \sigma_{R}\left(m\right)\right),
\]
which contradicts \eqref{eq:se}.

Hence, there exists a cutoff type \(t^* \in T^{\overline{\sigma}}_{\overline{m}}\) such that for any \(t \in T^{\overline{\sigma}}_{\overline{m}}\) and \(t \geq t^*\), we have \(t \in T^{\overline{\sigma}\geq \sigma}_{\overline{m}}\), while for any \(t \in T^{\overline{\sigma}}_{\overline{m}}\) and \(t < t^*\), we have \(t \in T^{\overline{\sigma} < \sigma}_{\overline{m}}\).

Next, we show that the interpretation of \(\overline{m}\) in \(\sigma\) rather than in \(\sigma\) can trigger an unraveling under the original ranking, i.e., when \(f\) is the identity function. We prove by contradiction. Suppose not and the unraveling condition \eqref{eq:unraveling} is violated at some type \(j \in T^{\overline{\sigma} < \sigma}_{\overline{m}}\) at some message \(m_{j}\). Observe that \(j < n\). Following the notation in Definition~\ref{def:unraveling}, we have
\[
u_{S}\left(j, \overline{\sigma}\right) < u_{S}\left(j,m_{j},\text{BR}\left(m_{j},\mu'\right)\right) = u_{S}\left(j, m_{j}, \text{BR}\left(m_{j},T^{\sigma}_{m_{j}} \cap \left\{t \leq j\right\}\right)\right).
\]
The equality follows from the fact that for any \(t > j\), we have either \(t \in T^{\overline{\sigma}\geq \sigma}_{\overline{m}}\) or \(t \in F^{\overline{\sigma} < \sigma}_{\overline{m}} \left(j\right)\). Then, we look at the \(j\)-truncated game \(G_{\text{S}}^{j}\) where the sender's type space \(T^{j}\) is \(\left\{1, 2, \dots, j\right\}\) and the prior distribution \(p^j\) is the \(T^{j}\)-conditional belief \(p_{T^{j}}\). We claim that there exists an \(j\)-truncated equilibrium \(\sigma^{j} \in \text{PSE}\left(G_{\text{S}}^{j}\right)\) such that
\[
u_{S}\left(j, \sigma^{j}\right) \geq u_{S}\left(j,m_{j},\text{BR}\left(m_{j},\mu'\right)\right) > u_{S}\left(j, \overline{\sigma}\right),
\]
which leads to a contradiction with the fact that \(\overline{\sigma}\) is the LMSE. To show the claim, we consider the following two cases.

Case 1: Type 1 does not pool with type \(j\) in the equilibrium \(\sigma\). Let \(k^* \in T^{j}\) be the highest type who does not pool with type \(j\) in the equilibrium \(\sigma\). Then \(1 \leq k^* < j\) and a \(k^*\)-truncated equilibrium \(\sigma^{k^*}\) can be derived from \(\sigma\) by deleting those types higher than \(k^*\) and keeping everything else unchanged. Let \(m^{k}_{k}\) denote the message sent by type \(k\) at the \(k\)-truncated equilibrium \(\sigma^{k}\). By Lemma~\ref{lemma:monotone-message}, we have \(m^{k^*}_{k^*} < m_{j}\). For any \(k^* \leq k < j\), let \(I_{k} = \left\{\left. t \in T^{k}\right| \sigma^{k}_{S}\left(t\right) = m^{k}_{k}\right\}\) and \(J_{k} = \{k+1, \dots, j\}\). Then, \(\sigma^{k^*}_{R}\left(m^{k^*}_{k^*}\right) = \text{BR}\left(m^{k^*}_{k^*}, I_{k^*}\right)\) and \(\text{BR}\left(m_{j},\mu'\right) = \text{BR}\left(m_{j}, J_{k^*}\right)\). We use induction to show that given the existence of a \(k\)-truncated equilibrium \(\sigma^{k}\) such that
\[
u_{S}\left(t, \sigma^{k}\right) \geq u_{S}\left(t, m_{j}, \text{BR}\left(m_{j}, J_{k^*}\right)\right) \quad \forall t \in T^{k}
\]
and \(m^{k}_{k}< m_j\), we can either directly construct a \(j\)-truncated equilibrium \(\sigma^j\) that satisfies our claim or indirectly construct a \(k+1\)-truncated equilibrium \(\sigma^{k+1}\) such that
\[
u_{S}\left(t, \sigma^{k+1}\right) \geq u_{S}\left(t, m_{j}, \text{BR}\left(m_{j}, J_{k^*}\right)\right) \quad \forall t \in T^{k+1}
\]
and \(m^{k+1}_{k+1}< m_j\). When \(k+1 = j\), we also prove our claim. Given the definition of \(\sigma^{k^*}\) and the equilibrium condition of \(\sigma\), we know that
\[
u_{S}\left(t, \sigma^{k^*}\right) = u_{S}\left(t, \sigma\right) \geq u_{S}\left(t, m_{j},\sigma_{R}\left(m_{j}\right)\right) \geq u_{S}\left(t, m_{j}, \text{BR}\left(m_{j}, J_{k^*}\right)\right) \quad \forall t \in T^{k^*}.
\]
For any \(k^* \leq k < j\), we consider type \(k+1\). By definition, type \(k+1\) pools with type \(j\) at \(\sigma\), i.e., \(\sigma_{S}\left(k+1\right) = m_{j}\). There are two cases:
\begin{itemize}
    \item Case 1.1: \(u_{S}\left(k+1, m^{k}_{k}, \text{BR}\left(m^{k}_{k}, I_{k}\right)\right) < u_{S}\left(k+1, m_{j}, \text{BR}\left(m_{j}, J_{k}\right)\right)\). Now we construct a \(j\)-truncated equilibrium \(\sigma^{j}\) based on \(\sigma^{k}\) by pooling types higher than \(k\) together. Since \(k \geq k^*\) and \(J_{k}\)-conditional belief first-order stochastic dominates \(J_{k^*}\)-conditional belief, we have
    \[
    u_{S}\left(k+1, m_{j}, \text{BR}\left(m_{j}, J_{k}\right)\right) \geq u_{S}\left(k+1, m_{j}, \text{BR}\left(m_{j}, J_{k^*}\right)\right).
    \]
    By the equilibrium condition of \(\sigma\) and A\ref{assumption:extreme-messages}, we have
    \begin{align*}
        & u_{S}\left(k+1, \sigma\right) = u_{S}\left(k+1, m_{j}, \text{BR}\left(m_{j}, J_{k^*}\right)\right) \\
        \geq & u_{S}\left(k+1, m^{l}, \sigma_{R}\left(m^{l}\right)\right) \geq u_{S}\left(k+1, m^{l}, \text{BR}\left(m^{l}, \left\{1\right\}\right)\right) \\
        > & u_{S}\left(k+1, m^{h}, \text{BR}\left(m^{h}, \left\{n\right\}\right)\right) \geq u_{S}\left(k+1, m^{h}, \text{BR}\left(m^{h}, J_{k}\right)\right)
    \end{align*}
    The above inequalities are due to the fact that the \(\left\{n\right\}\)-conditional belief first-order stochastic dominates \(J_{k}\)-conditional belief, and \(\left\{1\right\}\)-conditional belief is the worst belief. By the same argument, we have
    \[
    u_{S}\left(k+1, m^{k}_{k}, \text{BR}\left(m^{k}_{k}, I_{k}\right)\right) > u_{S}\left(k+1, m^{h}, \text{BR}\left(m^{h}, J_{k}\right)\right).
    \]
    Because \(u_{S}\left(k+1, m, \text{BR}\left(m, J_{k}\right)\right)\) is continuous in \(m\), there exists a message \(m^{j}_{j} \in \left[m_{j}, m^{h}\right)\) such that
    \begin{align*}
        & u_{S}\left(k+1, m^{j}_{j}, \text{BR}\left(m^{j}_{j}, J_{k}\right)\right) \\
            = &\max\left\{u_{S}\left(k+1, m_{j}, \text{BR}\left(m_{j}, J_{k^*}\right)\right), u_{S}\left(k+1, m^{k}_{k}, \text{BR}\left(m^{k}_{k}, I_{k}\right)\right)\right\},
    \end{align*}
    where \(m^{j}_{j} = m_{j}\) if and only if \(k = k^*\).
    
    The strategy-belief profile \(\sigma^{j}\) in the \(j\)-truncated game \(G^{j}_{S}\) is constructed based on the \(k\)-truncated equilibrium \(\sigma^{k}\) as follows:
    \begin{itemize}
        \item \(\forall t < k+1\), \(\sigma^{j}_{S}\left(t\right) = \sigma^{k}_{S}\left(t\right)\);
        \item \(\forall t \in J_{k}\), \(\sigma^{j}_{S}\left(t\right) = m_{j}^{j}\);
        \item \(\forall m \leq m^{k}_{k}\), \(\mu^{j}\left(\left.\cdot\right|m\right) = \mu^{k}\left(\left.\cdot\right|m\right)\) and \(\sigma^{j}_{R}\left(m\right) = \sigma^{k}_{R}\left(m\right)\).
        \item \(\mu^{j}\left(\left.\cdot\right|m^{j}_{j}\right) = p_{J_{k}}\) and \(\sigma^{j}_{R}\left(m^{j}_{j}\right) = \text{BR}\left(m^{j}_{j}, J_{k}\right)\).
        \item \(\forall m > m^{k}_{k}\) and \(m \neq m^{j}_{j}\), \(\mu^{j}\left(\left.\cdot\right|m\right) = p_{\left\{1\right\}}\) and \(\sigma^{j}_{R}\left(m\right) = \text{BR}\left(m, \left\{1\right\}\right)\).
    \end{itemize}
    Next we check that \(\sigma^{j}\) is a \(j\)-truncated equilibrium:
    \begin{itemize}
        \item \(t < k+1\):

        If \(m \leq m^{k}_{k}\), by the equilibrium condition of \(\sigma^{k}\) and the definition of \(\sigma^{j}\), we have
        \begin{align*}
                & u_{S}\left(t, \sigma^{j}\right) = u_{S}\left(t, \sigma^{k}\right)\\
            \geq & u_{S}\left(t, m, \text{BR}\left(m, \mu^{k}\right)\right) = u_{S}\left(t, m, \text{BR}\left(m, \mu^{j}\right)\right).
        \end{align*}
        If \(m = m^{j}_{j}\), by the equilibrium condition of \(\sigma^{k}\) and the reverse single-crossing condition (\(m^{k}_{k} <m_j \leq m_{j}^{j}\)), we have either
        \begin{align*}
            & u_{S}\left(t, \sigma^{j}\right) = u_{S}\left(t, \sigma^{k}\right) \\
            \geq & u_{S}\left(t, m_{j}, \text{BR}\left(m_{j}, J_{k^*}\right)\right) \geq u_{S}\left(t, m^{j}_{j}, \text{BR}\left(m^{j}_{j}, J_{k}\right)\right)
        \end{align*}
        when \(u_{S}\left(k+1, m^{j}_{j}, \text{BR}\left(m^{j}_{j}, J_{k}\right)\right) = u_{S}\left(k+1, m_{j}, \text{BR}\left(m_{j}, J_{k^*}\right)\right)\) or 
        \begin{align*}
            & u_{S}\left(t, \sigma^{j}\right) = u_{S}\left(t, \sigma^{k}\right) \\
            \geq & u_{S}\left(t, m^{k}_{k}, \text{BR}\left(m^{k}_{k}, I_{k}\right)\right) \geq u_{S}\left(t, m^{j}_{j}, \text{BR}\left(m^{j}_{j}, J_{k}\right)\right)
        \end{align*}
        when \(u_{S}\left(k+1, m^{j}_{j}, \text{BR}\left(m^{j}_{j}, J_{k}\right)\right) = u_{S}\left(k+1, m^{k}_{k}, \text{BR}\left(m^{k}_{k}, I_{k}\right)\right)\).

        If \(m > m^{k}_{k}\) and \(m \neq m^{j}_{j}\), by the equilibrium condition of \(\sigma^{k}\) and monotonicity (A\ref{assumption:monotonicity}), we have
        \begin{align*}
            & u_{S}\left(t, \sigma^{j}\right) = u_{S}\left(t, \sigma^{k}\right) \\
            \geq & u_{S}\left(t, m, \text{BR}\left(m, \mu^{k}\right)\right) \geq u_{S}\left(t, m, \text{BR}\left(m, \left\{1\right\}\right)\right).
        \end{align*}
        \item \(t \in J_{k}\):
        
        If \(m \leq m^{k}_{k}\), by the definition of \(m^{j}_{j}\) and the single-crossing condition, we have
        \begin{align*}
            & u_{S}\left(t, \sigma^{j}\right) = u_{S}\left(t, m^{j}_{j}, \text{BR}\left(m^{j}_{j}, J_{k}\right)\right) \\
            \geq & u_{S}\left(t, m^{k}_{k}, \text{BR}\left(m^{k}_{k}, I_{k}\right)\right)\\
            \geq & u_{S}\left(t, m, \text{BR}\left(m, \mu^{k}\right)\right) = u_{S}\left(t, m, \text{BR}\left(m, \mu^{j}\right)\right).
        \end{align*}

        If \(m > m^{k}_{k}\) and \(m \neq m^{j}_{j}\), by the equilibrium condition of \(\sigma^{k}\) and monotonicity, we have
        \begin{align*}
            & u_{S}\left(t, \sigma^{j}\right) = u_{S}\left(t, m^{j}_{j}, \text{BR}\left(m^{j}_{j}, J_{k}\right)\right) \\
            \geq & u_{S}\left(t, m^{k}_{k}, \text{BR}\left(m^{k}_{k}, I_{k}\right)\right)\\
            \geq & u_{S}\left(t, m^{l}, \text{BR}\left(m^{l}, \left\{1\right\}\right)\right) \geq u_{S}\left(t, m, \text{BR}\left(m, \left\{1\right\}\right)\right).
        \end{align*}
    \end{itemize}
    Hence, \(\sigma^{j}\) is a \(j\)-truncated equilibrium. Notice that by construction, we have
    \begin{align*}
        u_{S}\left(t, \sigma^{j}\right) &= u_{S}\left(t, \sigma^{k}\right) \geq u_{S}\left(t, m_{j}, \text{BR}\left(m_{j}, J_{k^*}\right)\right) \quad \forall t < k+1\\
        u_{S}\left(t, \sigma^{j}\right) &\geq u_{S}\left(t, m_{j}, \text{BR}\left(m_{j}, J_{k^*}\right)\right) \quad \forall t \in J_{k}.
    \end{align*}
    In particular, we have \(u_{S}\left(t, \sigma^{j}\right) \geq u_{S}\left(t, m_{j}, \text{BR}\left(m_{j}, \mu'\right)\right)\).
    \item Case 1.2: \(u_{S}\left(k+1, m^{k}_{k}, \text{BR}\left(m^{k}_{k}, I_{k}\right)\right) \geq u_{S}\left(k+1, m_{j}, \text{BR}\left(m_{j}, J_{k}\right)\right)\). Now we only need to show that there exists a \(k+1\)-truncated equilibrium
    \(\sigma^{k+1}\) such that
    \[
    u_{S}\left(t, \sigma^{k+1}\right) \geq u_{S}\left(t, m_{j}, \text{BR}\left(m_{j}, J_{k^*}\right)\right) \quad \forall t \in T^{k+1}
    \]
    and \(m^{k+1}_{k+1} < m_{j}\). If \(k+1 = j\), we prove the claim. Otherwise, we can replace \(k+1\), \(I_{k}\) and \(J_{k}\) by \(k+2\), \(I_{k+1} = I_{k} \cup \left\{k+1\right\}\) and \(J_{k+1}= J_{k} \setminus\left\{k+1\right\}\) respectively, repeat the same analysis for type \(k+2\) and so on until we reach type \(j\).
    
    We construct a \(k+1\)-truncated equilibrium \(\sigma^{k+1}\) based on the \(k\)-truncated equilibrium \(\sigma^{k}\) by letting type \(k+1\) pool with type \(k\). Denote \(l = \min\left\{t \in I_{k}\right\}\). Notice that
    \begin{align*}
        u_{S}\left(l, m^{k}_{k}, \text{BR}\left(m^{k}_{k}, I_{k+1}\right)\right) & >  u_{S}\left(l, m^{k}_{k}, \text{BR}\left(m^{k}_{k}, I_{k}\right)\right) = u_{S}\left(l, \sigma^{k}\right) \\
        u_{S}\left(l, m_{j}, \text{BR}\left(m_{j}, I_{k+1}\right)\right) & < u_{S}\left(l, m_{j}, \text{BR}\left(m_{j}, J_{k^*}\right)\right) \leq u_{S}\left(l, \sigma^{k}\right)
    \end{align*}
    The first inequality follows from the \(I_{k+1}\)-conditional belief first-order stochastically dominating the \(I_{k}\)-conditional belief. The second inequality follows from the \(J_{k^*}\)-conditional belief first-order stochastically dominating the \(I_{k+1}\)-conditional belief. The last equality follows from the property of \(\sigma^{k}\).

    Because \(u_{S}\left(l, m, \text{BR}\left(m, I_{k+1}\right)\right)\) is continuous in \(m\), there exists a \(m^{k+1}_{k+1} \in \left( m^{k}_{k}, m_{j}\right) \) such that
    \[
    u_{S}\left(l, m^{k+1}_{k+1}, \text{BR}\left(m^{k+1}_{k+1}, I_{k+1}\right)\right) = u_{S}\left(l, \sigma^{k}\right)
    \]
    The strategy-belief profile \(\sigma^{k+1}\) in the \(k+1\)-truncated game \(G_{\text{S}}^{k+1}\) is constructed based on the \(k\)-truncated equilibrium \(\sigma^{k}\) as follows:
    \begin{itemize}
        \item \(\forall t < l\), \(\sigma^{k+1}_{S}\left(t\right) = \sigma^{k}_{S}\left(t\right)\);
        \item \(\forall t \in I_{k+1}\), \(\sigma^{k+1}_{S}\left(t\right) = m^{k+1}_{k+1}\);
        \item \(\forall m < m^{k+1}_{k+1}\), \(\mu^{k+1}\left(\left.\cdot\right|m\right) = \mu^{k}\left(\left.\cdot\right|m\right)\) and \(\sigma^{k+1}_{R}\left(m\right) = \sigma^{k}_{R}\left(m\right)\).
        \item \(\mu^{k+1}\left(\left.\cdot\right|m^{k+1}_{k+1}\right) = p_{I_{k+1}}\) and \(\sigma^{k+1}_{R}\left(m^{k+1}_{k+1}\right) = \text{BR}\left(m^{k+1}_{k+1}, I_{k+1}\right)\).
        \item \(\forall m > m^{k+1}_{k+1}\), \(\mu^{k+1}\left(\left.\cdot\right|m\right) = p_{\left\{1\right\}}\) and \(\sigma^{k+1}_{R}\left(m\right) = \text{BR}\left(m, \left\{1\right\}\right)\).
    \end{itemize}
    Next we check that \(\sigma^{k+1}\) is a \(k+1\)-truncated equilibrium:
    \begin{itemize}
        \item \(t < l\):

        If \(m < m^{k+1}_{k+1}\), by the equilibrium condition of \(\sigma^{k}\) and the definition of \(\sigma^{k+1}\), we have
        \begin{align*}
                & u_{S}\left(t, \sigma^{k+1}\right) = u_{S}\left(t, \sigma^{k}\right)\\
            \geq & u_{S}\left(t, m, \text{BR}\left(m, \mu^{k}\right)\right) = u_{S}\left(t, m, \text{BR}\left(m, \mu^{k+1}\right)\right).
        \end{align*}
        If \(m = m^{k+1}_{k+1}\), by the equilibrium condition of \(\sigma^{k}\) and the reverse single-crossing condition (\(m^{k}_{k}<m^{k+1}_{k+1}\)), we have
        \begin{align*}
            & u_{S}\left(t, \sigma^{k+1}\right) \geq u_{S}\left(t, m^{k}_{k}, \text{BR}\left(m^{k}_{k}, I_{k}\right)\right) \\
            > & u_{S}\left(t, m^{k+1}_{k+1}, \text{BR}\left(m^{k+1}_{k+1}, I_{k+1}\right)\right).
        \end{align*}
        If \(m > m^{k+1}_{k+1}\), by the equilibrium condition of \(\sigma^{k}\) and monotonicity, we have
        \begin{align*}
            & u_{S}\left(t, \sigma^{k+1}\right) = u_{S}\left(t, \sigma^{k}\right)\\
            \geq & u_{S}\left(t, m, \text{BR}\left(m, \mu^{k}\right)\right) \geq u_{S}\left(t, m, \text{BR}\left(m, \left\{1\right\}\right)\right).
        \end{align*}
        \item \(t \in I_{k}\):
        
        If \(m < m^{k+1}_{k+1}\), by the equilibrium condition of \(\sigma^{k}\) and the single-crossing condition, we have
        \begin{align*}
            & u_{S}\left(t, \sigma^{k+1}\right) = u_{S}\left(t, m^{k+1}_{k+1}, \text{BR}\left(m^{k+1}_{k+1}, I_{k+1}\right)\right) \\
            \geq & u_{S}\left(t, m^{k}_{k}, \text{BR}\left(m^{k}_{k}, I_{k}\right)\right) = u_{S}\left(t, \sigma^{k}\right)\\
            \geq & u_{S}\left(t, m, \text{BR}\left(m, \mu^{k}\right)\right) = u_{S}\left(t, m, \text{BR}\left(m, \mu^{k+1}\right)\right).
        \end{align*}

        If \(m > m^{k+1}_{k+1}\), by the equilibrium condition of \(\sigma^{k}\) and monotonicity, we have
        \begin{align*}
            & u_{S}\left(t, \sigma^{k+1}\right) \geq u_{S}\left(t, \sigma^{k}\right)\\
            \geq & u_{S}\left(t, m, \text{BR}\left(m, \mu^{k}\right)\right) \geq u_{S}\left(t, m, \text{BR}\left(m, \left\{1\right\}\right)\right).
        \end{align*}

        \item \(t = k+1\):

        If \(m < m^{k+1}_{k+1}\), by the single-crossing condition and the previous results for \(t \in I_{k}\), we have
        \begin{align*}
            & u_{S}\left(k+1, \sigma^{k+1}\right) = u_{S}\left(k+1, m^{k+1}_{k+1}, \text{BR}\left(m^{k+1}_{k+1}, I_{k+1}\right)\right) \\
            \geq & u_{S}\left(k+1, m, \text{BR}\left(m, \mu^{k}\right)\right) = u_{S}\left(k+1, m, \text{BR}\left(m, \mu^{k+1}\right)\right).
        \end{align*}
        If \(m > m^{k+1}_{k+1}\), by the equilibrium condition of \(\sigma^{k}\) and the single-crossing condition (\(m^{k+1}_{k+1} > m^{l}\)), we have
        \begin{align*}
            & u_{S}\left(k+1, \sigma^{k+1}\right) = u_{S}\left(k+1, m^{k+1}_{k+1}, \text{BR}\left(m^{k+1}_{k+1}, I_{k+1}\right)\right) \\
            > & u_{S}\left(k+1, m^{l}, \text{BR}\left(m^{l}, \left\{1\right\}\right)\right) \geq u_{S}\left(k+1, m, \text{BR}\left(m, \left\{1\right\}\right)\right).
        \end{align*}
    \end{itemize}
    Hence, \(\sigma^{k+1}\) is a \(k+1\)-truncated equilibrium. Notice that by construction, we have
    \begin{align*}
        u_{S}\left(t, \sigma^{k+1}\right) &= u_{S}\left(t, \sigma^{k}\right) \geq u_{S}\left(t, m_{j}, \text{BR}\left(m_{j}, J_{k^*}\right)\right) \quad \forall t < l\\
        u_{S}\left(t, \sigma^{k+1}\right) &= u_{S}\left(t, m^{k+1}_{k+1}, \text{BR}\left(m^{k+1}_{k+1}, I_{k+1}\right)\right) \\
        &\geq u_{S}\left(t, m_{j}, \text{BR}\left(m_{j}, J_{k^*}\right)\right) \quad \forall t \in I_{k+1}.
    \end{align*}
    The second inequality follows from the fact that
    \begin{align*}
        u_{S}\left(l, \sigma^{k}\right) &\geq u_{S}\left(l, m_{j}, \text{BR}\left(m_{j}, J_{k^*}\right)\right) \\
        u_{S}\left(l, \sigma^{k}\right) &= u_{S}\left(l, m^{k+1}_{k+1}, \text{BR}\left(m^{k+1}_{k+1}, I_{k+1}\right)\right)
    \end{align*}
    and the single-crossing condition.

    In particular, we have \(u_{S}\left(k+1, \sigma^{k+1}\right) \geq u_{S}\left(t, m_{j}, \text{BR}\left(m_{j}, \mu'\right)\right)\).
    If \(k+1 = j\), we prove the claim. Otherwise, if
    \begin{align*}
    & u_{S}\left(k+2, m^{k+1}_{k+1}, \text{BR}\left(m^{k+1}_{k+1}, I_{k+1}\right)\right) \\
    \geq &u_{S}\left(k+2, m_{j}, \text{BR}\left(m_{j}, J_{k+1}\right)\right),
    \end{align*}
    we move to Case 1.2 and construct a \(k+2\)-truncated equilibrium based on \(\sigma^{k+1}\) as before.
    If
    \begin{align*}
    & u_{S}\left(k+2, m^{k+1}_{k+1}, \text{BR}\left(m^{k+1}_{k+1}, I_{k+1}\right)\right) \\
    < &u_{S}\left(k+2, m_{j}, \text{BR}\left(m_{j}, J_{k+1}\right)\right),
    \end{align*}
    we move to Case 1.1 and directly construct a \(j\)-truncated equilibrium \(\sigma^{j}\) based on \(\sigma^{k+1}\) as before.
\end{itemize}
Case 2: Type 1 pools with type \(j\) in the equilibrium \(\sigma\). Notice that we can construct a \(\left\{1\right\}\)-truncated equilibrium \(\sigma^{1}\) by letting the sender send the message \(m^{l}\) and assigning the \(\left\{1\right\}\)-conditional belief to the receiver for every message.
\begin{itemize}
    \item Case 2.1: \(u_{S}\left(1, m^{l}, \text{BR}\left(m^{l}, \left\{1\right\}\right)\right) > u_{S}\left(1, m_{j}, \text{BR}\left(m_{j}, \mu'\right)\right)\). Then \(m^{l} < m_{j}\), and we are in the same situation as Case 1 where \(k = 1\), \(I_{1} = \left\{1\right\}\), and \(J_{1} = \left\{2, \dots, j\right\}\). Hence, we can construct a \(j\)-truncated equilibrium \(\sigma^{j}\) based on \(\sigma^{1}\). 
    \item Case 2.2: \(u_{S}\left(1, m^{l}, \text{BR}\left(m^{l}, \left\{1\right\}\right)\right) \leq u_{S}\left(1, m_{j}, \text{BR}\left(m_{j}, \mu'\right)\right)\). Then we can directly construct a strategy-belief profile \(\sigma^{j}\) by pooling all types together as follows:
    \begin{itemize}
        \item \(\forall t \in T^{j}\), \(\sigma^{j}_{S}\left(t\right) = m_{j}\);
        \item \(\mu^{j}\left(\left.\cdot\right|m_{j}\right) = p_{T^{j}}\) and \(\sigma^{j}_{R}\left(m_{j}\right) = \text{BR}\left(m_{j}, T^{j}\right) = \text{BR}\left(m_{j}, \mu'\right)\).
        \item \(\forall m \neq m_{j}\), \(\mu^{j}\left(\left.\cdot\right|m\right) = p_{\left\{1\right\}}\) and \(\sigma^{j}_{R}\left(m\right) = \text{BR}\left(m, \left\{1\right\}\right)\).
    \end{itemize}
    We check that \(\sigma^{j}\) is a \(j\)-truncated equilibrium: for all \(t \in T^{j}\) and all \(m \neq m_{j}\), we have
    \[
    u_{S}\left(t, m_{j}, \text{BR}\left(m_{j}, T^{j}\right)\right) \geq u_{S}\left(t, m^{l}, \text{BR}\left(m^{l}, \left\{1\right\}\right)\right) \geq u_{S}\left(t, m, \text{BR}\left(m, \left\{1\right\}\right)\right),
    \]
    which follows from the single-crossing condition and A\ref{assumption:extreme-messages}. In particular, we have \(u_{S}\left(j, \sigma^{j}\right) \geq u_{S}\left(j, m_{j}, \text{BR}\left(m_{j}, \mu'\right)\right)\).
\end{itemize}

We have shown that if the unraveling condition \eqref{eq:unraveling} is violated at type \(j\), then there exists a \(j\)-truncated equilibrium \(\sigma^{j}\) such that
\[
u_{S}\left(j, \sigma^{j}\right) \geq u_{S}\left(j, m_{j}, \text{BR}\left(m_{j}, \mu'\right)\right) > u_{S}\left(j, \overline{\sigma}\right).
\]
However, this contradicts the fact that \(\overline{\sigma}\) is the LMSE (Lemma~\ref{lemma:LMSE}). Therefore, we conclude that the interpretation of \(\overline{m}\) in \(\overline{\sigma}\) and \(\sigma\) rather than triggers an unraveling.

For the LMSE \(\overline{\sigma}\) and any other equilibrium \(\sigma\), we have shown the existence of a message \(\overline{m}\) such that the interpretation of \(\overline{m}\) in \(\overline{\sigma}\) rather than \(\sigma\) triggers an unraveling. Hence, \(\overline{\sigma}\) is most persuasive.

\subsection{Proof of Theorem~\ref{thm:utility-uniqueness}}\label{thm:proof-utility-uniqueness}

We prove by contradiction. Suppose that there exists another most persuasive equilibrium \(\hat{\sigma}\) in the game \(G_{\text{S}}\) that is not payoff-equivalent for the sender to the LMSE \(\overline{\sigma}\) (Definition~\ref{def:payoff-equivalence}). Then, by the definition of most persuasive equilibrium (Definition~\ref{def:most-persuasive}), there exists a message \(\hat{m}\) such that the interpretation of \(\hat{m}\) in \(\hat{\sigma}\) rather than \(\overline{\sigma}\) triggers an unraveling. In particular, there exists a type \( i \in T\) who sends \(\hat{m}\) in the equilibrium \(\hat{\sigma}\) and strictly prefers \(\hat{\sigma}\) to \(\tilde{\sigma}\), i.e., \(u_{S}\left(i, \hat{\sigma}\right) > u_{S}\left(i, \overline{\sigma}\right)\). Now we show that the interpretation of \(\hat{m}\) in \(\hat{\sigma}\) rather than \(\overline{\sigma}\) can never trigger an unraveling irrespective of the ranking function \(f\), which contradicts the fact that \(\hat{\sigma}\) is most persuasive.
    
The first observation is that type \(i\) must pool with higher types in the equilibrium \(\hat{\sigma}\). Otherwise, \(\hat{\sigma}\) induces a \(i\)-truncated equilibrium \(\hat{\sigma}^{i}\) where \(u_{S}\left(i, \hat{\sigma}^{i}\right) > u_{S}\left(i, \overline{\sigma}\right)\), which contradicts the fact that \(\overline{\sigma}\) is the LMSE (Lemma~\ref{lemma:LMSE}). Secondly, if we denote by \(j > i\) the highest type pooling with type \(i\) in the equilibrium \(\hat{\sigma}\), then there must exist \( i < k \leq j \) such that \(u_{S}\left(k, \hat{\sigma}\right) < u_{S}\left(k, \overline{\sigma}\right)\). Let \(\overline{m}_{k}\) denote the message sent by type \(k\) in the LMSE \(\overline{\sigma}\). Following the notation of Definition~\ref{def:unraveling}, we replace \(\sigma\) and \(\overline{\sigma}\) by \(\overline{\sigma}\) and \(\hat{\sigma}\) respectively. For any ranking function \(f\), let
    \[
    t^{*} =\arg \max_{t \in T^{\hat{\sigma} < \overline{\sigma}}_{\hat{m}} \cap T^{\overline{\sigma}}_{\overline{m}_{k}}}  f\left(t\right).
    \]
Then, by the definition of \(t^{*}\), we have \(F^{\hat{\sigma} < \overline{\sigma}}_{\hat{m}}\left(t^{*}\right) \cap T^{\overline{\sigma}}_{\overline{m}_{k}} = \emptyset\). Our goal is to show that
\[
u_{S}\left(t^{*},\hat{\sigma}\right) < u_{S}\left(t^{*},\overline{m}_{k},\text{BR}\left(\overline{m}_{k},\overline{\mu}^*\right)\right) = u_{S}\left(t^{*},\overline{m}_{k},\text{BR}\left(\overline{m}_{k},U^{\overline{\sigma}}_{\overline{m}_{k}}\right)\right) .
\]
In other words, the interpretation of \(\hat{m}\) in \(\hat{\sigma}\) rather than \(\overline{\sigma}\) cannot trigger an unraveling, and the unraveling process has an early stop when reaching type \(t^{*}\).

Notice that
\begin{align}
    u_{S}\left(i, \hat{m}, \hat{\sigma}_{R}\left(\hat{m}\right)\right) = u_{S}\left(i, \hat{\sigma}\right) &> u_{S}\left(i, \overline{\sigma}\right) \geq u_{S}\left(i, \overline{m}_{k}, \overline{\sigma}_{R}\left(\overline{m}_{k}\right)\right) \label{eq:uniqueness-1} \\
    u_{S}\left(k, \hat{m}, \hat{\sigma}_{R}\left(\hat{m}\right)\right) = u_{S}\left(k, \hat{\sigma}\right) &< u_{S}\left(k, \overline{\sigma}\right) = u_{S}\left(k, \overline{m}_{k}, \overline{\sigma}_{R}\left(\overline{m}_{k}\right)\right) \label{eq:uniqueness-2}
\end{align}
First, we claim that \(\hat{m} < \overline{m}_{k}\). Otherwise, if \(\hat{m} = \overline{m}_{k}\), then by monotonicity, \eqref{eq:uniqueness-1} implies that the receiver must take a higher action in \(\hat{\sigma}\) than in \(\overline{\sigma}\) after observing the messages \(\hat{m}\) and \(\overline{m}_{k}\) respectively, which contradicts \eqref{eq:uniqueness-2}; if \(\hat{m} > \overline{m}_{k}\), by the single-crossing condition, \eqref{eq:uniqueness-1} implies that \(u_{S}\left(k, \hat{\sigma}\right) > u_{S}\left(k, \overline{\sigma}\right)\), which contradicts \eqref{eq:uniqueness-2} again. Hence, we have \(\hat{m} < \overline{m}_{k}\).

Next we claim that \(\min \left\{t \in T^{\hat{\sigma}}_{\hat{m}}\right\} \leq \min \left\{t \in T^{\overline{\sigma}}_{\overline{m}_k}\right\}\). Otherwise, \(i \in T^{\overline{\sigma}}_{\overline{m}_k}\). Given that \(u_{S}\left(k, \hat{\sigma}\right) < u_{S}\left(k, \overline{\sigma}\right)\) and \(u_{S}\left(i, \hat{\sigma}\right) > u_{S}\left(i, \overline{\sigma}\right)\), by the single-crossing condition, there exists \( i \leq i' < k\) such that: (1) for any type \(t \in T^{\hat{\sigma}}_{\hat{m}}\) such that \(t \leq i'\), we have \(u_{S}\left(t, \hat{\sigma}\right) \geq u_{S}\left(t, \overline{\sigma}\right)\); (2) for any type \(t \in T^{\hat{\sigma}}_{\hat{m}}\) such that \(t > i'\), we have \(u_{S}\left(t, \hat{\sigma}\right) < u_{S}\left(t, \overline{m}_k, \overline{\sigma}_{R}\left(\overline{m}_k\right)\right) \leq u_{S}\left(t, \overline{\sigma}\right)\). Let \(i'' = \max \left\{\left. t \in T \right| t < \min \left\{t \in T^{\hat{\sigma}}_{\hat{m}}\right\}\right\}\). We have
\[
    u_{S}\left(i'', \hat{\sigma}\right) \geq u_{S}\left(i'', \hat{m}, \hat{\sigma}_{R}\left(\hat{m}\right)\right) > u_{S}\left(i'', \overline{m}_k, \overline{\sigma}_{R}\left(\overline{m}_k\right)\right)= u_{S}\left(i'', \overline{\sigma}\right).
\]
Then, \(\hat{\sigma}\) induces a \(i''\)-truncated equilibrium \(\hat{\sigma}^{i''}\) where \(u_{S}\left(i'', \hat{\sigma}^{i''}\right) > u_{S}\left(i'', \overline{\sigma}\right)\), which contradicts the fact that \(\overline{\sigma}\) is the LMSE (Lemma~\ref{lemma:LMSE}).

Given that \(u_{S}\left(k, \hat{\sigma}\right) < u_{S}\left(k, \overline{\sigma}\right)\) and \(\min \left\{t \in T^{\hat{\sigma}}_{\hat{m}}\right\} \leq \min \left\{t \in T^{\overline{\sigma}}_{\overline{m}_k}\right\}\), by the single-crossing condition, there exists a cutoff type \(k' < k\) such that for any type \(t \in T^{\hat{\sigma}}_{\hat{m}} \cap T^{\overline{\sigma}}_{\overline{m}_{k}}\) and \(t \leq k'\), we have \(u_{S}\left(t, \hat{\sigma}\right) \geq u_{S}\left(t, \overline{\sigma}\right)\).\footnote{\(k' < \min \left\{t \in T^{\overline{\sigma}}_{\overline{m}_k}\right\}\) implies that \(u_{S}\left(t, \hat{\sigma}\right) < u_{S}\left(t, \overline{\sigma}\right)\) for any type \(t \in T^{\hat{\sigma}}_{\hat{m}} \cap T^{\overline{\sigma}}_{\overline{m}_{k}}\).} Then,
\begin{align*}
    T^{\hat{\sigma} \geq \overline{\sigma}}_{\hat{m}}  & = T^{\hat{\sigma} \geq \overline{\sigma}}_{\hat{m}} \cap \left\{\left.t \in T\right|t \leq k'\right\} \\
    U^{\overline{\sigma}}_{\overline{m}_{k}} & = T^{\overline{\sigma}}_{\overline{m}_{k}} \setminus \left(F^{\hat{\sigma} < \overline{\sigma}}_{\hat{m}}\left(t^{*}\right) \cup T^{\hat{\sigma} \geq \overline{\sigma}}_{\hat{m}}\right) = T^{\overline{\sigma}}_{\overline{m}_{k}} \cap \left\{\left.t \in T\right|t > k'\right\} \\
    u_{S}\left(t^{*}, \hat{\sigma}\right) & < u_{S}\left(t^*, \overline{\sigma}\right) \leq u_{S}\left(t^{*}, \overline{m}_{k}, \text{BR}\left(\overline{m}_{k}, U^{\overline{\sigma}}_{\overline{m}_{k}}\right)\right),
\end{align*}
which implies that the unraveling process stops early when reaching type \(t^{*}\).

Therefore, we have shown that the interpretation of \(\hat{m}\) in \(\hat{\sigma}\) rather than \(\overline{\sigma}\) can never trigger an unraveling irrespective of the ranking function, which contradicts the fact that \(\hat{\sigma}\) is most persuasive. Hence, the most persuasive equilibrium is unique up to payoff equivalence for the sender, which is determined by the LMSE.

\subsection{Proof of Theorem~\ref{thm:outcome-uniqueness}}\label{thm:proof-outcome-uniqueness}

Suppose that there are two most persuasive equilibria \(\hat{\sigma}\) and \(\overline{\sigma}\) in the game \(G_{\text{S}}\). By Theorem~\ref{thm:utility-uniqueness}, we know that they are payoff-equivalent for the sender, i.e., \(u_{S}\left(t, \hat{\sigma}\right) = u_{S}\left(t, \overline{\sigma}\right)\) for all \(t \in T\), and they are both LMSE. To show the uniqueness of the most persuasive equilibrium outcome, we only need to show that \(\hat{\sigma}_{S}\left(t\right) = \overline{\sigma}_{S}\left(t\right)\) for all \(t \in T\). When the sender's equilibrium strategy is pinned down, so does the equilibrium outcome.\footnote{\(\hat{\sigma}\) and \(\overline{\sigma}\) can differ in terms of the off-path beliefs and what happens following an off-path message.}

Notice that \(\hat{\sigma}\) and \(\overline{\sigma}\) must induce the same partition of the type space. If not, perturbing the prior will affect the posterior after each message, and thus the equality in payoffs will vanish, which implies that the most persuasive equilibrium outcome is generically unique. Now we show that the sender of the same type must send the same message in both equilibria. Suppose not. If \(\hat{\sigma}_{S}\left(j\right) = \hat{m}_{j} \neq \overline{m}_{j} = \overline{\sigma}_{S}\left(j\right)\) for some \(j \in T\), then 
\[
u_{S}\left(j, \hat{m}_{j}, \hat{\sigma}_{R}\left(\hat{m}_{j}\right)\right) = u_{S}\left(j, \overline{m}_{j}, \overline{\sigma}_{R}\left(\overline{m}_{j}\right)\right)
\]
implies that
\[
u_{S}\left(i, \hat{m}_{j}, \hat{\sigma}_{R}\left(\hat{m}_{j}\right)\right) \neq u_{S}\left(i, \overline{m}_{j}, \overline{\sigma}_{R}\left(\overline{m}_{j}\right)\right)
\]
if type \(j\) pools with any type \(i \neq j\) by the single-crossing condition, which contradicts the assumption that both equilibria generate the same equilibrium payoff for type \(i\). Hence, type \(j\) does not pool with any other type in both equilibria.

Consider the \(j\)-truncated game \(G_{\text{S}}^{j}\). Since type \(j\) does not pool with any other type, \(\overline{\sigma}\) and \(\hat{\sigma}\) induce two \(j\)-truncated equilibria \(\overline{\sigma}^{j}\) and \(\hat{\sigma}^{j}\). Strict quasi-concavity implies that \(u_{S}\left(j, m, \text{BR}\left(m, \left\{j\right\}\right)\right) > u_{S}\left(j, \overline{\sigma}^{j}\right) = u_{S}\left(j, \hat{\sigma}^{j}\right)\) for all \(m\) between \(\overline{m}_{j}\) and \(\hat{m}_{j}\). 

If there is no type lower than \(j\), i.e., \(j = 1\), then we can construct another \(j\)-truncated equilibrium \(\tilde{\sigma}^{j}\) by letting type \(j\) sending some message between \(\overline{m}_{j}\) and \(\hat{m}_{j}\). Then, type \(j\) achieves a higher payoff in \(\tilde{\sigma}^{j}\) than in \(\overline{\sigma}^{j}\), contradicting the assumption that \(\overline{\sigma}\) is the LMSE.

If there exist types lower than \(j\), then
\[
u_{S}\left(j-1, \overline{m}_{j}, \overline{\sigma}^{j}_{R}\left(\overline{m}_{j}\right)\right) \neq u_{S}\left(j-1, \hat{m}_{j}, \hat{\sigma}^{j}_{R}\left(\hat{m}_{j}\right)\right).
\]
Otherwise, the single-cross condition implies that \(u_{S}\left(j, \overline{\sigma}^{j}\right) \neq u_{S}\left(j, \hat{\sigma}^{j}\right)\), which contradicts our assumption. Since \(u_{S}\left(j-1, \overline{\sigma}^{j}\right) = u_{S}\left(j-1, \hat{\sigma}^{j}\right)\), we have either
\[
u_{S}\left(j-1, \overline{\sigma}^{j}\right) > u_{S}\left(j-1, \overline{m}_{j}, \overline{\sigma}^{j}_{R}\left(\overline{m}_{j}\right)\right)
\]
or
\[
u_{S}\left(j-1,\hat{\sigma}^{j}\right) > u_{S}\left(j-1, \hat{m}_{j}, \hat{\sigma}^{j}_{R}\left(\hat{m}_{j}\right)\right).
\]
Without loss of generality, we can take \(\overline{\sigma}^{j}\) for example. Given that
\[
u_{S}\left(j, m, \text{BR}\left(m, \left\{j\right\}\right)\right) > u_{S}\left(j, \overline{\sigma}^{j}\right)
\]
for all \(m\) between \(\overline{m}_{j}\) and \(\hat{m}_{j}\), continuity implies that there exists a message \(\tilde{m}_{j}\) close to \(\overline{m}_{j}\) such that
\begin{align*}
    u_{S}\left(j, \tilde{m}_{j}, \text{BR}\left(\tilde{m}_{j}, \left\{j\right\}\right)\right) &> u_{S}\left(j, \overline{\sigma}^{j}\right) \\
    u_{S}\left(j-1, \tilde{m}_{j}, \text{BR}\left(\tilde{m}_{j}, \left\{j\right\}\right)\right) &< u_{S}\left(j-1, \overline{\sigma}^{j}\right).
\end{align*}
By the single-crossing condition, we have \(u_{S}\left(t, \overline{\sigma}^{j}\right) > u_{S}\left(t, \tilde{m}_{j}, \text{BR}\left(\tilde{m}_{j}, \left\{j\right\}\right)\right)\) for all \(t < j\). Then, we can construct a \(j\)-truncated equilibrium \(\tilde{\sigma}^{j}\) based on \(\overline{\sigma}^{j}\) by letting \(\tilde{\sigma}^j\left(t\right) = \overline{\sigma}^{j}\left(t\right)\) for \(t < j\), and \(\tilde{\sigma}^j\left(j\right) = \tilde{m}_{j}\). Then, type \(j\) achieves a higher payoff in \(\tilde{\sigma}^{j}\) than in \(\overline{\sigma}^{j}\), contradicting the assumption that \(\overline{\sigma}\) is the LMSE.

Hence, the sender of the same type must send the same message in the two most persuasive equilibria \(\overline{\sigma}\) and \(\hat{\sigma}\), and they produce a unique lex max outcome, which is generically the unique most persuasive equilibrium outcome.

\section{Intuitive Explanation of the Criteria in Table~\ref{tab:comparison}}\label{sec:refinements}

We use the intuitive criterion as the benchmark and compare it with other refinements. We follow the two-step approach in Section~\ref{sec:persuasive}. The explanations are not meant to provide exact characterizations but to convey the essential intuition underlying each criterion. For ease of comparison, we begin by restating the intuitive criterion.

\subsection*{Intuitive Criterion \citep{choSignalingGamesStable1987}}
\begin{itemize}
    \item Step 1: Which types of the sender \emph{could benefit} by sending an off-path message \(m\)?
    
    We denote the set of such types as \(D\). Formally,
    \[
    D = \left\{\left. t \in T \right| u_{S}\left(t, \sigma\right) \leq \max_{a \in \text{BR}\left(m, \Delta\left(T\right)\right)} u_{S}\left(t, m, a\right)\right\},
    \]
    where \(\text{BR}\left(m, \Delta\left(T\right)\right) = \cup_{\mu \in \Delta\left(T\right)}\text{BR}\left(m, \mu\right)\).

    \item Step 2: If deviations only come from the set of types of the sender identified in Step 1, is the \emph{lowest} payoff from deviating higher than their equilibrium payoff for some type of the sender?

    Formally, if there exists \(t \in D\) such that
    \[
    \min_{a \in \text{BR}\left(m, \Delta\left(D\right)\right)} u_{S}\left(t, m, a\right) > u_{S}\left(t, \sigma\right),
    \]
    then this equilibrium \(\sigma\) fails the intuitive criterion.
\end{itemize}

\subsection*{D1 Criterion \citep{banksEquilibriumSelectionSignaling1987}}
\begin{itemize}
    \item Step 1: Which types of the sender  \emph{are most likely to benefit} by sending an off-path message \(m\)?
    \item Step 2: The same as Step 2 of the intuitive criterion.
\end{itemize}

Type \(t_1\) is more likely to benefit than type \(t_2\) by sending an off-path message \(m\), if whenever type \(t_2\) can weakly benefit by sending \(m\) under some action of the receiver, type \(t_1\) can strictly benefit by sending \(m\) under the same action. In the Spencian game, when the high-type worker pools with the medium-type worker in an equilibrium, the high-type worker is more likely to benefit by sending an off-path message than the medium-type worker. Then, in Step 2, we only consider the higher-type worker when checking whether the equilibrium fails the D1 criterion, even though the medium-type worker could also benefit from deviating. Hence, the D1 criterion is stronger than the intuitive criterion, and it uniquely selects the Riley outcome in monotone signaling games.\footnote{This explanation follows from \citet{munoz-garciaIntuitiveDivinityCriterion2011}.}

\subsection*{G-P Criterion \citep{grossmanPerfectSequentialEquilibrium1986}}

\begin{itemize}
		\item Step 1: The same as Step 1 of the intuitive criterion.
		\item Step 2: If deviations only come from a \emph{subset} \(\tilde{D} \subseteq D\) of the types of the sender identified in Step 1, is the \emph{expected} payoff from deviating higher than their equilibrium payoff \emph{only for} \(t \in \tilde{D}\)? Formally, if there exists \(\tilde{D} \subseteq D\) such that
		\begin{align*}
			u_{S}\left(t, m, \text{BR}\left(m, p_{\tilde{D}}\right)\right) &\geq u_{S}\left(t, \sigma\right) \quad \forall t \in \tilde{D}\\
			u_{S}\left(\tilde{t}, m, \text{BR}\left(m, p_{\tilde{D}}\right)\right) &> u_{S}\left(\tilde{t}, \sigma\right) \quad \exists \tilde{t} \in \tilde{D}\\
			u_{S}\left(t, m, \text{BR}\left(m, p_{\tilde{D}}\right)\right) &\leq u_{S}\left(t, \sigma\right) \quad \forall t \in T \setminus \tilde{D}
		\end{align*}
		where \(p_{\tilde{D}}\left(t\right) = p\left(\left.t\right|t \in \tilde{D}\right)\) is the \(\tilde{D}\)-conditional belief, then the equilibrium fails the G-P Criterion.
\end{itemize}

The G-P criterion is stronger than the intuitive criterion, because the belief applied in Step 2 is less pessimistic than that in the intuitive criterion. As a result, the G-P criterion uniquely select \(\sigma^{m_2}\) in Example~\ref{ex:hiding}. However, the G-P criterion suffers from the non-existence problem. In Example~\ref{ex:stiglitz}, the intuitive criterion selects the Riley outcome irrespective of the prior probability \(p\). Then, the only equilibrium outcome that could pass the G-P criterion is the Riley outcome. However, when \(p\) is close to zero, the Riley outcome fails the G-P criterion, because both the high-type and low-type workers could benefit from deviating to an education level close to zero, thereby obtaining an expected wage of \(2-p\).

\subsection*{Undefeated Equilibrium \citep*{mailathBeliefBasedRefinementsSignalling1993}}

\(\sigma'\) is defeated by \(\sigma\) if (1) there exists an \emph{off-path} message \(m\) in \(\sigma'\) that is \emph{on-path} in \(\sigma\) such that for \emph{every type} \(t\) that sends \(m\) in \(\sigma\), the \emph{equilibrium} payoff in \(\sigma\) is weakly higher than in \(\sigma'\); and (2) there exists one type \(t'\) sending \(m\) in \(\sigma\) whose equilibrium payoff is strictly higher than in \(\sigma'\).

An equilibrium is undefeated if it is not defeated by any other equilibrium. The undefeated equilibrium is typically not unique. In Example~\ref{ex:undefeated}, both the pooling equilibrium \(\sigma^{\text{Pooling}}\) and the LMSE \(\overline{\sigma}\) are undefeated. Neither the pooling equilibrium nor the LMSE defeats the other, because the medium-type worker prefers the pooling equilibrium, while the high-type worker prefers the LMSE. Similarly, in Example~\ref{ex:beer-quiche}, both pooling equilibria \(\sigma^{\text{Beer}}\) and \(\sigma^{\text{Quiche}}\) are undefeated, because the surly type prefers \(\sigma^{\text{Beer}}\), while the wimp type prefers \(\sigma^{\text{Quiche}}\).

\end{document}